\documentclass[acmsmall, screen]{acmart}
\usepackage{dashbox}
\usepackage{framed}
\usepackage{cleveref}
\usepackage{fancyvrb}
\usepackage{lipsum}
\usepackage{xspace}
\usepackage{mathpartir}
\usepackage{ltl}
\usepackage{tikz}
\usepackage{enumitem}
\usetikzlibrary{automata, positioning, arrows, backgrounds, shapes.geometric}
\tikzset{->,
    shorten >=1pt, auto, node distance=1.5cm, on grid, initial text=,
    every state/.style={minimum size=20pt,inner sep=0pt},
    every node/.style={font=\normalsize}}
\tikzstyle{mytriangle}=[isosceles triangle, shape border rotate=-180, fill=gray!25, xshift=1cm]

\theoremstyle{definition}
\newtheorem{axiom}{Axiom}

\newcommand{\iosteps}[3]{#1 \overset{#2}{\rightsquigarrow} #3}

\newcommand{\functor}[1]{#1 \xrightarrow{\mathit{mon}} #1}
\newcommand{\tracerel}[2]{\mathcal{P}(#1^\omega \times #2^\omega)}
\newcommand{\quadruple}[4]{#1 \vdash {#2}^\forall \mid {#3}^\exists \ \{ \ #4 \ \}}
\newcommand{\iquadruple}[4]{#1 \vdash_\texttt{inv} {#2}^\forall \mid {#3}^\exists \ \{ \ #4 \ \}}

\title{Coinductive Proofs for Temporal Hyperliveness}

\setcopyright{rightsretained}
\acmDOI{10.1145/3704889}
\acmYear{2025}
\acmJournal{PACMPL}
\acmVolume{9}
\acmNumber{POPL}
\acmArticle{53}
\acmMonth{1}
\received{2024-07-11}
\received[accepted]{2024-11-07}

\begin{document}

  \author{Arthur Correnson}
  \orcid{0000-0003-2307-2296}
  \affiliation{%
    \institution{CISPA Helmholtz Center for Information Security}
    \city{Saarbruecken}
    \country{Germany}
  }
  \email{arthur.correnson@cispa.de}

  \author{Bernd Finkbeiner}
  \orcid{0000-0002-4280-8441}
  \affiliation{%
    \institution{CISPA Helmholtz Center for Information Security}
    \city{Saarbruecken}
    \country{Germany}
  }
  \email{finkbeiner@cispa.de}

  \begin{abstract}
    Temporal logics for hyperproperties have recently emerged as an
    expressive specification technique for relational properties of
    reactive systems.
    While the model checking problem for such logics has been widely studied,
    there is a scarcity of deductive proof systems for temporal hyperproperties.
    In particular, hyperproperties with an alternation of universal and existential quantification over system executions are rarely supported.
    In this paper, we focus on hyperproperties of the form $\forall^*\exists^*\psi$, where $\psi$ is a safety relation.
    We show that hyperproperties of this class -- which includes many hyperliveness properties of interest -- can always be approximated by coinductive relations. This enables intuitive proofs by coinduction.
    Based on this observation, we define \textsc{HyCo} (\textbf{Hy}perproperties, \textbf{Co}inductively), a mechanized framework to reason about temporal hyperproperties within the Coq proof assistant.
    We detail the construction of HyCo, provide a proof of its soundness, and exemplify its use by applying it to the verification of reactive systems modeled as imperative programs with nondeterminism and I/O.
  \end{abstract}



\maketitle

  \section{Introduction}

 Temporal logics for hyperproperties like HyperLTL~\cite{DBLP:conf/post/ClarksonFKMRS14} combine
temporal reasoning over infinite executions with the ability to
quantify universally and existentially over multiple 
executions. Such logics are ideally suited to verify reactive systems,
where the continuous interaction of the system with its environment results in an infinite sequence of events. In this context, temporal logics for hyperproperties enable the specification of information-flow security policies as well as other relational properties like robustness, knowledge, and causality.
For example, the information-flow policy of \emph{noninference}~\cite{McLean:1994:GeneralTheory} requires that an observer of a system should
not be able to observe any difference in the behavior of the system if we replace its high-security inputs by a dummy value~$\lambda$. Thus, from the point of view of a potential attacker, sensitive data manipulated by the system are indistinguishable from $\lambda$.
Expressing this requirement formally requires quantification over several executions of the system.
In HyperLTL, noninference is expressed as the formula
\[ \forall \tau_1 \, \exists \tau_2 \, \LTLsquare\, (hi_{\tau_2} = \lambda \wedge lo_{\tau_1} = lo_{\tau_2} ) \]
where $hi$ denotes the high-security inputs and $lo$ the low-security outputs of the system.
This formula requires for every execution the existence of another execution where the high-security inputs to the system are replaced by a dummy value~$\lambda$ and yet the low-security outputs remain the same. Noninference is an example of \textit{hyperliveness} property \cite{4556678}.

There has been a lot of recent research on finite-state and
abstraction-based model checking for temporal
hyperproperities~\cite{FRS15,DBLP:conf/tacas/HsuSB21,DBLP:conf/cav/BeutnerF22,DBLP:conf/tacas/BeutnerF23,DBLP:conf/tacas/Beutner24}. Somewhat surprisingly, however, there is a scarcity of deductive proof systems supporting logics like HyperLTL. In particular, contrary to more traditional Hoare-style program logics, there exists, to the best of our knowledge, no mechanized framework to reason about temporal hyperproperties within a proof assistant.
This is unfortunate, because the automated approaches are still very limited in
their scalability; the verification of realistic systems appears far
out of reach.

A fragment that is often supported by deductive approaches is the
special case of
$k$-hypersafety~\cite{AntonopoulosGHK17,ShemerGSV19,FarzanV19,FarzanV20},
where all executions are quantified universally. This case is
attractive, because the verification of a $k$-hypersafety property can
be reduced to the verification of a trace property of the $k$-fold
selfcomposition of the given system.

The general case is much more difficult. Hyperliveness properties like noninference contain an alternation between universal and existential quantification. The $\forall\exists$ pattern occurs in many hyperproperties where the existential quantifier resolves the potential nondeterminism in the system. To prove such hyperproperties, we must, for every
execution, establish the existence of another execution such that the
two executions respect a specified temporal relation. A
straight-forward idea is to provide an explicit witness for the
existential quantifier (cf.~\cite{DBLP:conf/csfw/LamportS21}).
One way to do this is to view the interaction between the universal and existential quantifiers as an infinite game, such that a winning strategy for the existential player provides a witness for the existential quantifier.
For finite-state systems, the winning strategy can be found using techniques from reactive
synthesis~\cite{DBLP:conf/cav/CoenenFST19}.
In the deductive setting, guessing a winning strategy upfront is not only difficult; any error in the initial choice of
strategy might furthermore only be revealed very late in the proof,
requiring a restart of the proof from the very beginning.  A practical proof
technique should therefore be \emph{incremental}, refining the instantiation of the existential
quantifier on-the-fly as needed during the proof.

In this paper, we present such an incremental proof
technique. We accomplish this by linking temporal hyperproperties to
coinductive relations. Coinductive relations are supported by
parameterized coinduction \cite{paco}, a powerful proof technique 
with built-in support for incremental proofs.
This coalgebraic view precisely coincides with the
game-theoretic interpretation in the sense that they can prove exactly
the same properties. However, the coalgebraic approach is much better suited for an interactive proof assistant.
In summary, we make the following contributions:

 \begin{enumerate}[leftmargin=*, itemsep=5pt]
    \item We introduce \textsc{HyCo}, a coinductive relation to reason about the difficult class of hyperproperties of the form $\forall^*\exists^*\psi$, where $\psi$ is a safety relation between traces.
    Importantly, \textsc{HyCo} is language agnostic: it does not commit to a specific programming language to model systems, nor to a particular logic to specify trace relations.

    \item We formally prove the soundness of HyCo in the Coq proof assistant.
    Our proof relies on a subtle use of the axiom of functional choice to resolve the underlying existential quantification.

    \item We further equip \textsc{HyCo} with a collection of incremental reasoning rules, useful up-to techniques, and specialized rules for the case where properties are specified using LTL-style temporal modalities.
    The soundness of the rules for temporal modalities critically relies on the 
    notion of \textit{derivatives of trace relations}, combined with the idea to reason \textit{up to stronger trace relations}.

    \item Finally, we show that \textsc{HyCo} applies naturally to the verification of reactive systems modeled as imperative programs with I/O.
    We provide specialized reasoning rules for this case, and give examples of program proofs using HyCo in Coq.
  \end{enumerate}
 
  \section{Preliminaries}

  \subsection{Infinite Sequences and Set Notations}

  Given an alphabet $\Sigma$, we note $\Sigma^\omega$ the set of all infinite sequences of letters in $\Sigma$. Given an infinite sequence $\tau \in \Sigma^\omega$, we note $\tau[i]$ ($i \in \mathbb{N}$) the letter at position $i$ in $\tau$, and we note $\tau[i...]$ the infinite trace obtained by ignoring the first $i$ letters of $\tau$. In particular, if $\tau = \tau_0\tau_1\tau_2...$ we have $\tau[0] = \tau_0$ and $\tau[1...] = \tau_1\tau_2...$.
  Given two sets $A$ and $B$, note $A \equiv B$ if $A$ and $B$ are mutually included.

  \subsection{Labeled Transition Systems}

  Throughout this paper, we model reactive systems as labeled transition systems (LTS). Formally, a LTS is a triple $(\mathcal{S}, \mathcal{E}, \mathcal{I}, \to)$ where $\mathcal{S}$ is an arbitrary set of states, $\mathcal{E}$ is an arbitrary set of \textit{observable events}, $\mathcal{I} \in \mathcal{S}$ is an \textit{initial} state, and $\ \to \ \subseteq \mathcal{S} \times \mathcal{E} \cup \{ \emptyset \} \times \mathcal{S}$ is a labeled \textit{transition relation} between two states and an (optional) event. Given two states $s_1, s_2 \in \mathcal{S}$, the relation $s_1 \xrightarrow{\emptyset} s_2$ indicates that $s_1$ can transition to $s_2$ without emitting any event. When $e \in \mathcal{E}$, $s_1 \xrightarrow{e} s_2$ indicates that $s_1$ can transition to $s_2$ while emitting the observable event $e$.
  We also adopt the following notation conventions: \begin{align*}
    \to \quad &\textrm{stands for} \xrightarrow{\emptyset}\\
    \to^* \quad &\textrm{is the reflexive transitive closure of $\to$}\\
    \overset{e}{\rightsquigarrow} \quad &\textrm{is a shorthand for} \to^* \xrightarrow{e}
  \end{align*}

  \paragraph*{Trace Semantics}

  Given a LTS $\mathit{TS} = (\mathcal{S}, \mathcal{E}, \mathcal{I}, \to)$, a trace originating from a state $s \in \mathcal{S}$ is an infinite sequence of events $\tau \in \mathcal{E}^\omega$ such that there exists an infinite sequence of states $\pi \in \mathcal{S}^\omega$ with \[
    \pi_0 = s \wedge \forall i, \iosteps{\pi_i}{\tau_i}{\pi_{i + 1}}
  \]

  Given a state $s \subseteq \mathcal{S}$, we note $\mathit{Traces}_{\mathit{TS}}(s) \subseteq \mathcal{E}^\omega$ the set of traces originating from $s$.
  When the transition system being considered is clear from context, we simply note $\mathit{Traces}(s)$ for the traces of $s$.
  If $S \subseteq \mathcal{S}$ we note $\mathit{Traces}_\mathit{TS}(S) \triangleq \bigcup_{s \in S} \mathit{Traces}(s)$. If $\mathcal{I}$ is the initial state of $\mathit{TS}$, we note $\mathit{Traces}(\mathit{TS}) \triangleq \mathit{Traces}(\mathcal{I})$ the traces of $\mathit{TS}$.

  \subsection{Temporal Hyperproperties}

  \newcommand{\hltl}{$\textrm{Hyper}_E$\xspace}

  To specify hyperproperties of LTS, we introduce \hltl, an event-based variant of HyperLTL with support for arbitrary temporal relations between traces.
  A \hltl formula $\Psi$ starts with a sequence of $n$ quantifiers ranging over the traces of some LTS and then specifies a $m$-ary ($m \ge n$) trace relation $\psi$. \[
    \Psi ::= \forall^\mathit{TS}\,\Psi ~~\mid~~ \exists^\mathit{TS}\,\Psi ~~\mid~~ \psi\\
  \]

  Given an (ordered) $n$-tuple of traces $\Gamma$, we note $\Gamma \models \Psi$ when $\Gamma$ satisfies the relation specified by~$\Psi$. The model relation is defined as follows:
  \begin{align*}
    \Gamma \models \forall^\mathit{TS}\,\psi &\iff \forall \tau \in \mathit{Traces}(\mathit{TS}), (\Gamma, \tau) \models \psi\\
    \Gamma \models \exists^\mathit{TS}\, \psi &\iff \exists \tau \in \mathit{Traces}(\mathit{TS}), (\Gamma, \tau) \models \psi\\
    \Gamma \models \psi &\iff \Gamma \in \psi
  \end{align*}

  \noindent We will often consider relations $\psi$ expressed in (a fragment) of linear temporal logic.
  The basic assertions are $n$-ary relations $\varphi$ between events, and formulas are obtained by using boolean connectives, and the temporal modalities $\texttt{W}$ (weak-until), $\LTLsquare$ (always), and $\LTLnext$ (next).
  \[
    \psi ::= \varphi \mid \psi \wedge \psi \mid \psi \vee \psi \mid \psi \ \texttt{W} \ \psi \mid \LTLsquare \psi \mid \LTLnext \psi
  \]
  
  Given an $n$-tuple of traces $\Gamma = (\tau_1, ..., \tau_n)$, we note $\Gamma \models \psi$ when $\Gamma$ satisfies the temporal relation defined by $\psi$.
  We note $\Gamma[i]$ the $n$-tuple of events $(\tau_1[i], ..., \tau_n[i])$, and $\Gamma[i...]$ the $n$-tuple of traces $(\tau_1[i...], ..., \tau_n[i...])$.
  The interpretation of the modalities is defined as follows: \begin{align*}
    \Gamma \models \varphi &\iff \Gamma[0] \in \varphi\\
    \Gamma \models \LTLnext \psi &\iff \Gamma[1...] \models \psi\\
    \Gamma \models \LTLsquare \psi &\iff \forall i, \Gamma[i...] \models \psi\\
    \Gamma \models \psi_1 \ \texttt{W} \ \psi_2 &\iff (\forall i, \Gamma[i...] \models \psi_1) \vee
    (\exists i, \Gamma[i...] \models \psi_2 \wedge \forall j < i, \Gamma[j...] \models \psi_1)
  \end{align*}

  When $\psi$ is defined using combinations of LTL modalities, we freely interchange $\psi$ and its set of models $\{ \ \Gamma \mid \Gamma \models \psi \ \}$.

  \subsection{Parameterized Coinduction}

  We recall some basics about coinduction. Let $(\mathcal{P}(A), \subseteq, \cup)$ be the complete lattice of subsets of $A$, and let $F : \mathcal{P}(A) \xrightarrow{\mathit{mon}} \mathcal{P}(A)$ be a monotone operator on subsets. Then, $F$ has a greatest fixed point $\nu . F$ and \[
    \nu . F = \bigcup \{ \ R \mid R \subseteq F(R) \ \}
  \]

  An immediate consequence of this observation is that whenever we have a property $P$ defined as a greatest fixed point $P = \nu . F_P$ (i.e., defined \textit{coinductively}), we have a systematic method to show that an element $x$ satisfies $P$: it suffices to prove that $x$ is in some postfixed point $R$ of $F_P$. In this context, $R$ is usually called the \textit{coinduction hypothesis}.
  
  \begin{lemma}[Coinduction principle]
    \label{thm:coinduction}
    $X \subseteq \nu . F \iff \exists R, X \subseteq R \wedge R \subseteq F(R)$
  \end{lemma}

  A common use case for coinduction is proofs by simulation.
  As an example, let us consider the following two finite LTSs:

  \begin{center}
    \textbf{$TS_1$}:
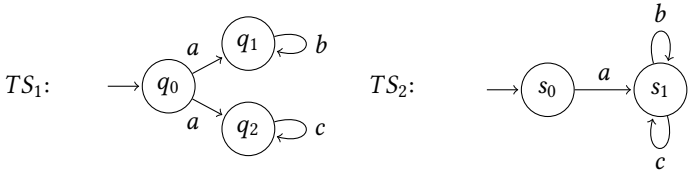
    \begin{minipage}{.3\textwidth}
      \centering
      \begin{tikzpicture}
        \node[initial, state] (A) {$q_0$};
        \node[state, above right = of A, yshift=-0.5cm] (B) {$q_1$};
        \node[state, below right = of A, yshift=+0.5cm] (C) {$q_2$};
        \draw[->] (A) edge node{$a$} (B);
        \draw[->] (A) edge[below left] node{$a$} (C);
        \draw[->] (B) edge[loop right] node{$b$} (B);
        \draw[->] (C) edge[loop right] node{$c$} (C);
      \end{tikzpicture}
    \end{minipage}%
    \textbf{$TS_2$}:
    \begin{minipage}{.3\textwidth}
      \centering
      \begin{tikzpicture}
        \node[initial, state] (A) {$s_0$};
        \node[state, right = of A] (B) {$s_1$};
        \draw[->] (A) edge node{$a$} (B);
        \draw[->] (B) edge[loop above] node{$b$} (B);
        \draw[->] (B) edge[loop below] node{$c$} (B);
      \end{tikzpicture}
    \end{minipage}
  \end{center}

  \noindent Clearly, all traces of $\mathit{TS}_1$ are also valid traces of $\mathit{TS}_2$ (i.e., $\mathit{TS}_1$ \textit{refines} $\mathit{TS}_2$). To prove this fact, one way is to enumerate the traces of $\mathit{TS}_1$, but this would be impossible if the number of traces was infinite.
  Another way it to reduce the reasoning about traces to reasoning about states by establishing a \textit{simulation relation}. The idea is to show the existence of a mapping $R \subseteq \mathcal{S}_1 \times \mathcal{S}_2$ between states of $\mathit{TS}_1$ and $\mathit{TS}_2$ such that $(q_0, s_0) \in R$, and such that for every pair of states $(q, s) \in R$, if $\iosteps{q}{e}{q'}$ then there exists a state $s'$ with $\iosteps{s}{e}{s'}$ and $(q', s')$ is again in $R$. The existence of such a mapping ensures that for all traces originating from $q_0$, there is a way to reproduce the same trace from $s_0$. This intuition can be understood in coalgebraic terms by observing that the existence of a simulation relation $R$ containing $(q_0, s_0)$ simply amounts to proving $(q_0, s_0) \in \texttt{sim}$ where $\texttt{sim}$ is the coinductive relation defined as follows: \begin{align*}
    \texttt{simF}(R) &\triangleq \{ (q, s) \mid \forall e . \forall q' . \iosteps{q}{e}{q'} \implies \exists s', \iosteps{s}{e}{s'} \wedge (q', s') \in R \}\\
    \texttt{sim} &\triangleq \nu . \texttt{simF}
  \end{align*}
  By \Cref{thm:coinduction}, $(q_0, s_0) \in \texttt{sim}$ if and only if there exists some relation $R$ containing $(q_0, s_0)$ and such that $R \subseteq \texttt{simF}(R)$. Going back to our simple example, $R = \{ (q_0, s_0), (q_1, s_1), (q_2, s_1) \}$ is a postfixed point of \texttt{simF}, which proves trace-inclusion.

  For our previous example, guessing a simulation $R$ was not too difficult. For larger examples, especially if we consider transition systems with infinite state-spaces, finding $R$ can be extremely tedious. Instead we could also start with $R = \{ (q_0, s_0) \}$, and explore the transition systems from there, carefully matching transitions in the left-hand transition system with corresponding transitions in the right-hand transition system and adding pairs of states to $R$ until we reach a postfixed point. To achieve this incremental proof style, Hur et al. proposed \textit{parameterized coinduction} \cite{paco}.
  The key idea is to replace the greatest fixed point $\nu . F$ with a parameterized variant $G_F(H)$ where $H$ is a current \textit{guess} for the coinduction hypothesis $R$.
  The \emph{parameterized greatest fixed point} operator is defined as \[
    G_F(H) \triangleq \nu . (\lambda X. F(X \cup H))
  \]
  By definition, we have $\nu . F \equiv G_F(\emptyset)$, which means that any time we need to prove a theorem of the form $X \subseteq \nu . F$, we can prove $X \subseteq G_F(\emptyset)$ instead.
  An important difference, however, is that $G_F$ can be equipped with the following set of \textit{incremental} reasoning rules:

  \setlength{\FrameSep}{1pt}
  \begin{figure}[ht!]
      \begin{mathpar}
        \inferrule[Init]{X \subseteq G_F(\emptyset)}{X \subseteq \nu F}
        \qquad
        \inferrule[Accumulate]{X \subseteq G_F(H \cup X)}{X \subseteq G_F(H)}
        \qquad
        \inferrule[Step]{X \subseteq F(H \cup G_F(H))}{X \subseteq G_F(H)}
      \end{mathpar}
      \vspace{-0.5cm}
      \caption{Rules of parameterized coinduction}
      \label{fig:coinduction}
      \Description{Rules of parameterized coinduction}
  \end{figure}

  Briefly,
  \begin{enumerate}
    \item The rule \textsc{Init} exploits the equation $\nu F \equiv G_F(\emptyset)$ to initiate a proof by parameterized coinduction.
    \item The rule \textsc{Accumulate} allows us to incrementally extend the current guess $H$
    \item The rule \textsc{Step} allows us to make progress in the proof by unfolding the fixed point equation $G_F(H) = F(H \cup G_F(H))$
  \end{enumerate}
  We refer the reader to \cite{paco} for a more detailed explanation of parameterized coinduction and for proofs of these reasoning rules.
  For now, let us go back to our simulation example, and give an incremental proof by parameterized coinduction.

  \begin{align*}
      & (q_0, s_0) \in \nu . \texttt{simF}\\
    \iff & (q_0, s_0) \in G_{\texttt{simF}}(\emptyset) & \textrm{\color{gray}(by \textsc{Init})}\\
    \Longleftarrow~& \underbrace{\exists s', \iosteps{s_0}{a}{s'} \wedge (q_1, s') \in G_{\texttt{simF}}(\emptyset)}_{\textrm{find a step matching $q_0 \xrightarrow{a} q_1$}} \wedge \underbrace{\exists s', \iosteps{s_0}{a}{s'} \wedge (q_2, s') \in G_{\texttt{simF}}(\emptyset)}_\textrm{find a step matching $q_0 \xrightarrow{a} q_2$} & \textrm{\color{gray}(by \textsc{Step})}\\
    \Longleftarrow~& (q_1, s_1) \in G_{\texttt{simF}}(\emptyset) \wedge (q_2, s_1) \in G_{\texttt{simF}}(\emptyset) & \textrm{\color{gray}(pick $s' = s_1$)}\\
    \Longleftarrow~& (q_1, s_1) \in G_{\texttt{simF}}(\{ (q_1, s_1) \}) \wedge (q_2, s_1) \in G_{\texttt{simF}}(\{ (q_2, s_1) \}) & \textrm{\color{gray}(by \textsc{Accumulate})}\\
    \Longleftarrow~& \underbrace{\exists s', \iosteps{s_1}{b}{s'} \wedge (q_1, s') \in \{ (q_1, s_1) \} \cup G_{\texttt{simF}}(\{ (q_1, s_1) \})}_\textrm{find a step matching $q_1 \xrightarrow{b} q_1$} \ \wedge & \textrm{\color{gray}(by \textsc{Step})}\\
    & \underbrace{\exists s', \iosteps{s_1}{c}{s'} \wedge (q_2, s') \in \{ (q_2, s_1) \} \cup G_{\texttt{simF}}(\{ (q_2, s_1) \})}_\textrm{find a step matching $q_2 \xrightarrow{c} q_2$}\\
    \Longleftarrow~& (q_1, s_1) \in \{ (q_1, s_1 ) \} \wedge (q_2, s_1) \in \{ (q_2, s_1 ) \} & \textrm{\color{gray}(pick $s' = s_1$)}
  \end{align*}

  \section{Hyperliveness and Relational Invariance, Coinductively}
  \label{sec:fei}

  \newcommand*{\feif}{\texttt{feiF}}
  \newcommand*{\fei}{\texttt{fei}}

  In the previous section, we recalled that trace inclusion between two transition systems can be reduced to reasoning about states by finding a simulation relation.
  Further, parameterized coinduction can be used to construct a simulation relation incrementally.
  Trace inclusion is an example of hyperliveness property: given two transition systems $\mathit{TS}_1$ and $\mathit{TS}_2$ emitting events in $\mathcal{E}$,
  it can be specified as the formula $\forall^{\mathit{TS}_1}\exists^{\mathit{TS}_2}\LTLsquare \mathit{eq}$ where $\mathit{eq} = \{ \ (e, e) \mid e \in \mathcal{E} \ \}$ is equality of events.
  In this section, we show that parameterized coinduction can also be used for incremental proofs of any hyperproperties of the form $\forall\exists\LTLsquare\varphi$ where $\varphi$ is a binary relation between events. To do so, we generalize the operator $\texttt{simF}$ and define a coinductive relation $\texttt{fei}_\varphi$ that soundly underapproximates $\forall\exists\LTLsquare\varphi$.
  We describe the construction of $\texttt{fei}_\varphi$, show that it can be equipped with useful reasoning rules, and propose a formal proof of soundness via the axiom of functional choice.
  In \Cref{sec:fe}, we extend this construction to support the more general class of $\forall\exists\psi$ hyperproperties where $\psi$ is a safety relation between traces.

  \subsection{Generalizing Simulation Relations}

  For the remainder of this section, let $\mathit{TS}_1 = (\mathcal{S}_1, \mathcal{E}_1, \mathcal{I}_1, \to_1)$ and $\mathit{TS}_2 = (\mathcal{S}_2, \mathcal{E}_2, \to_2)$ be two LTSs, and let us fix a binary relation between events $\varphi \subseteq \mathcal{E}_1 \times \mathcal{E}_2$.
  We define the relation $\fei_\varphi$ (\textbf{f}orall, \textbf{e}xists, \textbf{i}nvariant) as follows: \begin{align*}
    \feif_\varphi(R) &\triangleq \{ \ (s_1, s_2) \mid \forall \iosteps{s_1}{e_1}{s'_1} \implies \exists \iosteps{s_2}{e_2}{s'_2} \wedge (e_1, e_2) \in \varphi \wedge R(s'_1, s'_2) \ \}\\
    \fei_\varphi &\triangleq \nu . \feif_\varphi
  \end{align*}

  As expected, the $\fei_\varphi$ relation provides a sound proof technique to verify hyperproperties of the form $\forall\exists\LTLsquare\varphi$. More precisely, if the initial states of $\mathit{TS}_1$ and $\mathit{TS}_2$ are related by $\fei_\varphi$, it can then be concluded that $\forall^{\mathit{TS}_1}\exists^{\mathit{TS}_2}\LTLsquare \varphi$ is valid.

  \begin{theorem}[Soundness of $\fei_\varphi$]
    \label{thm:fei-sound}
    Let $\mathit{TS}_1$ and $\mathit{TS}_2$ be two LTSs and let $\varphi \subseteq \mathcal{E}_1 \times \mathcal{E}_2$. We have \[
      (\mathcal{I}_1, \mathcal{I}_2) \in \fei_\varphi \implies \models \forall^{\mathit{TS}_1}\exists^{\mathit{TS}_2}\LTLsquare \varphi
    \]
  \end{theorem}

  While this theorem might seem like an unsurprising generalization of simulation proofs, its formal proof turns out to be much more involved.
  Indeed, to prove this theorem we need to find a suitable trace to instantiate the inner existential quantifier. Note that in the simpler case of simulation proofs, since the two quantified traces are required to be equal, the existential quantifier is superfluous and it suffices to instantiate it with a copy of the universal trace. In the more general case of $\forall\exists\LTLsquare\varphi$, the witness for the existential quantifier depends on the relation $\varphi$ (and the quantified systems) and we need to somehow \textit{extract} it from the proof of $(\mathcal{I}_1, \mathcal{I}_2) \in \fei_\varphi$.
  We delay the explanation of this process to \Cref{sec:fei-proof}.

  \subsection{Incremental Proofs}

  As discussed in the preliminaries, a limitation of standard coinductive proofs is that they require us to find a coinduction hypothesis up-front. Instead, parameterized coinduction \cite{paco} allows us to perform coinductive proofs \textit{incrementally} by starting from $\emptyset$ and progressively
  guessing portions of a coinductive hypothesis until a postfixed point is reached.
  In this section, we leverage parameterized coinduction to prove hyperproperties incrementally via $\fei$.
  Taking inspiration from notations introduced in \cite{stuttering_for_free}, we start by defining the following two \textit{semantic quadruples}: \begin{align*}
    \iquadruple{\fbox{$H$}}{s_1}{s_2}{\varphi} &\triangleq (s_1, s_2) \in G_{\feif_\varphi}(H)\\
    \iquadruple{\dbox{$H$}}{s_1}{s_2}{\varphi} &\triangleq (s_1, s_2) \in H \cup G_{\feif_\varphi}(H)
  \end{align*}

  Note that $(s_1, s_2) \in \fei_\varphi \iff \iquadruple{\fbox{$\emptyset$}}{s_1}{s_2}{\varphi}$. In turn, to prove $\models \forall^{\mathit{TS}_1}\exists^{\mathit{TS}_2}\LTLsquare\varphi$, it suffices to show $\iquadruple{\fbox{$\emptyset$}}{s_1}{s_2}{\varphi}$. Further, since these quadruples are defined in terms of the parameterized greatest fixed point operator $G$, they admit incremental reasoning principles.
  The intention behind the choice of notation is that $\fbox{$H$}$ represents a \textit{guarded} coinduction hypothesis (i.e., it cannot yet be used to conclude a proof), whereas $\dbox{$H$}$ represents an \textit{unguarded} coinduction hypothesis (i.e., it can be used to immediately conclude a proof if $(s_1, s_2) \in H$).
  The core reasoning rules associated with these quadruples are presented in \Cref{fig:fei_core}. In the conclusion of some rules, we note $H^?$ instead of $\dbox{$H$}$ or $\fbox{$H$}$ when the rule can be applied regardless of whether the hypothesis is currently guarded or not.
  
  \begin{figure}[!ht]
      \begin{mathpar}
        \inferrule[Init]{ \iquadruple{\ \fbox{$\emptyset$}}{\mathcal{I}_1}{\mathcal{I}_2}{\varphi}}{ \models \forall^{\mathit{TS}_1}\exists^{\mathit{TS}_2} \LTLsquare \varphi } \qquad
        \inferrule[Cycle]{ (s_1, s_2) \in H }{ \iquadruple{\ \dbox{$H$}}{s_1}{s_2}{\varphi} }\\
        \inferrule[Step]{\iquadruple{\ \forall \iosteps{s_1}{e_1}{s'_1}, \ \exists \iosteps{s_2}{e_2}{s'_2}, (e_1, e_2) \in \varphi \wedge \dbox{$H$}}{s_1'}{s_2'}{\varphi} }{ \iquadruple{\ H^{?}}{s_1}{s_2}{\varphi} }\\
        \inferrule[Invariant]{ (s_1, s_2) \in H' \\ \forall (s_1, s_2) \in H', \iquadruple{\ \fbox{$H \cup H'$}}{s_1}{s_2}{\varphi} }{ \iquadruple{\ H^{?}}{s_1}{s_2}{\varphi} }
      \end{mathpar}
      \caption{Incremental reasoning rules for $\fei$}
      \Description{Incremental reasoning rules for $\fei$}
      \label{fig:fei_core}
  \end{figure}

  All rules are derived from the soundness of $\fei$ (\Cref{thm:fei-sound}), the rules of parameterized coinduction (\Cref{fig:coinduction}), and the definition of $\feif$.
  The rule \textsc{Init} exploits the soundness of $\fei$ to initiate a proof by parameterized coinduction.
  When focusing on a pair of states $(s_1, s_2)$, the rule \textsc{Step} requires to prove that for every transition $\iosteps{s_1}{e_1}{s'_1}$, there exists a corresponding transition $\iosteps{s_2}{e_2}{s'_2}$ such that $e_1$ and $e_2$ satisfy the event-invariant $\varphi$. Applying \textsc{Step} transfers the focus to the pairs of states $(s'_1, s'_2)$ and releases the guard.
  \textsc{Cycle} allows us to finish a proof by creating a cycle: if we already encountered a pair $(s_1, s_2)$ earlier in the proof (i.e., if $(s_1, s_2) \in H$) we can use this fact to immediately conclude. Note that applying \textsc{Cycle} requires the hypothesis to be unguarded.
  Finally, the rule $\textsc{Invariant}$ is a reformulation of the rule \textsc{Accumulate} of parameterized coinduction. It extends the current coinduction hypothesis $H$ with a larger guess $H \cup H'$, where $H'$ contains at least the current pair of states.
  The cost to pay is that applying \textsc{Invariant} restores the guard, and requires to prove a quadruple for every state in $H \cup H'$ instead of just the initial pair of states.
  The benefit is that, after applying the rule \textsc{Invariant}, we can conclude a proof as soon as we reach again a pair of states that satisfy the \textit{invariant} $H'$.
  We note that the rules \textsc{Step} and \textsc{Invariant} can be applied regardless of whether the hypothesis is currently guarded or not.

  \begin{example}
    Let us consider the following transition system, whose events are $\mathcal{E} = \{ a, b \}$:
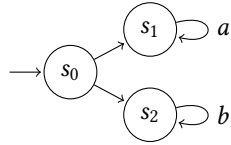
    \begin{center}
      \begin{tikzpicture}
        \node[initial, state] (A) {$s_0$};
        \node[state, above right = of A, yshift=-0.5cm] (B) {$s_1$};
        \node[state, below right = of A, yshift=+0.5cm] (C) {$s_2$};
        \draw[->] (A) edge (B);
        \draw[->] (A) edge (C);
        \draw[->] (B) edge[loop right] node{$a$} (B);
        \draw[->] (C) edge[loop right] node{$b$} (C);
      \end{tikzpicture}
    \end{center}

    \noindent Suppose we wish to prove $\models \forall^{\mathit{TS}}\exists^{\mathit{TS}}\LTLsquare(a_1 \leftrightarrow b_2)$ where $a_1 \leftrightarrow b_2$ stands for the binary relation $\{ \ (a, b), (b, a) \ \}$.
    We give an incremental proof using the rules of \Cref{fig:fei_core}.
    The proof starts by exploring the system until we reach states $s_1$ and $s_2$:
    \begin{align*}
      & \models \forall^\mathit{TS}\exists^\mathit{TS}\LTLsquare(a_1 \leftrightarrow b_2)\\
      \Longleftarrow \quad & \iquadruple{\fbox{$\emptyset$}}{s_0}{s_0}{a_1 \leftrightarrow b_2} & \textrm{\color{gray}(by \textsc{Init})}\\
      \Longleftarrow \quad & \iquadruple{\dbox{$\emptyset$}}{s_1}{s_2}{a_1 \leftrightarrow b_2} \wedge \iquadruple{\dbox{$\emptyset$}}{s_2}{s_1}{a_1 \leftrightarrow b_2} & \textrm{\color{gray}(by \textsc{Step})}
    \end{align*}
    Then, we observe that the only way out of the pair $(s_1, s_2)$ (resp. $(s_2, s_1)$) is to go back to $(s_1, s_2)$ (resp. $(s_2, s_1)$).
    This indicates that $H_1 = \{ (s_1, s_2) \}$ and $H_2 = \{ (s_2, s_1)\}$ are good choices of coinduction hypotheses. We therefore extend our current hypotheses with $H_1$ and $H_2$ using \textsc{Invariant}: \begin{align*}
      & \iquadruple{\dbox{$\emptyset$}}{s_1}{s_2}{a_1 \leftrightarrow b_2} \wedge \iquadruple{\dbox{$\emptyset$}}{s_2}{s_1}{a_1 \leftrightarrow b_2}\\
      \Longleftarrow \quad & \iquadruple{\fbox{$H_1$}}{s_1}{s_2}{a_1 \leftrightarrow b_2} \wedge \iquadruple{\fbox{$H_2$}}{s_2}{s_1}{a_1 \leftrightarrow b_2} & \textrm{\color{gray}(by \textsc{Invariant})}
    \end{align*}
    Using $\textsc{Step}$, $a$-transitions out of $s_1$ are matched with $b$-transitions out of $s_2$ (and vice-versa), thus maintaining the event-invariant $a_1 \leftrightarrow b_2$. Note that taking a step releases the guard! Since transitioning out of $(s_1, s_2)$ (resp. $(s_2, s_1)$) cycles back to $(s_1, s_2)$ (resp. $(s_2, s_1)$), we can then conclude with \textsc{Cycle}:
    \begin{align*}
      &\iquadruple{\fbox{$H_1$}}{s_1}{s_2}{a_1 \leftrightarrow b_2} \wedge \iquadruple{\fbox{$H_2$}}{s_2}{s_1}{a_1 \leftrightarrow b_2}\\
      \Longleftarrow \quad & \iquadruple{\dbox{$H_1$}}{s_1}{s_2}{a_1 \leftrightarrow b_2} \wedge \iquadruple{\dbox{$H_2$}}{s_2}{s_1}{a_1 \leftrightarrow b_2} & \textrm{\color{gray}(by \textsc{Step})}\\
      \Longleftarrow \quad & (s_1, s_2) \in H_1 = \{ (s_1, s_2) \} \wedge (s_2, s_1) \in H_2 = \{ (s_2, s_1) \} & \textrm{\color{gray}(by \textsc{Cycle})}
    \end{align*}
  \end{example}

  \subsection{Alignment Rules}
  \label{sec:align}

  Sometimes, it is convenient to be able to step through the execution from the left state and the right state at different paces in order to align the states in a certain way before establishing an invariant via \textsc{Invariant}.
  As an example, consider the following two transition systems:

  
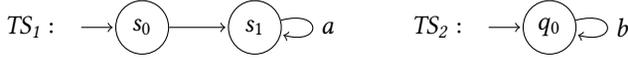
\begin{center}
    \begin{tikzpicture}
      \node (INIT) {$\mathit{TS_1}:$};
      \node[initial, state, right = of INIT] (A) {$s_0$};
      \node[state, right = of A] (B) {$s_1$};
      \draw[->] (A) edge (B);
      \draw[->] (B) edge[loop right] node{$a$} (B);
    \end{tikzpicture}
    \qquad
    \begin{tikzpicture}
      \node (INIT) {$\mathit{TS_2}:$};
      \node[initial, state, right = of INIT] (A) {$q_0$};
      \draw[->] (A) edge[loop right] node{$b$} (A);
    \end{tikzpicture}
  \end{center}

  Clearly, $\models \forall^{\mathit{TS}_1}\exists^{\mathit{TS}_2}\LTLsquare(a_1 \wedge b_2)$.
  One way to prove it is to use the \textsc{Step} rule to reach $(s_1, q_0)$, then \textsc{Invariant} to add $(s_1, q_0)$ to the current coinduction hypothesis, and conclude by using \textsc{Step} and \textsc{Cycle}. While this proof is perfectly valid, it forces us to use the \textsc{Step} rule two times to prove that an $a$ on the left-hand transition system can always be matched by a $b$ in the right-hand system. For such a small example, duplicating the reasoning is not difficult, but for more complicated systems, it might actually be extremely tedious. Instead, what we would like to do is to first align the states $s_1$ and $q_0$ by taking one step of computation in the left-hand transition system, then add $(s_1, q_0)$ to the current hypothesis, and finish the proof with a single application of the rule \textsc{Step} followed by \textsc{Cycle}.
  In this section, we introduce reasoning rules to support this intuition.

  To reproduce the proof we just described, we would need a rule of the following form: \begin{mathpar}
    \inferrule{ s_1 \to^* s'_1 \\ \iquadruple{\fbox{$H$}}{s'_1}{s_2}{\varphi} }{ \iquadruple{\fbox{$H$}}{s_1}{s_2}{\varphi} }
  \end{mathpar}

  This rule would allow us to skip silent computation steps in the left-hand transition system. Unfortunately, it is unsound in the presence of nondeterminism.
  Indeed, if along the path from $s_1$ to $s'_1$, there are possibilities to branch out to some other state $s''_1$, this rule would not ensure that the step to $s''_1$ can be matched by a corresponding step from $s_2$.
  A potential solution would be to \textit{determinize} the rule by universally quantifying on the post-states of $s_1$: \begin{mathpar}
    \inferrule{\forall s'_1, s_1 \to^* s'_1 \implies \iquadruple{\fbox{$H$}}{s'_1}{ s_2}{\varphi} }{ \iquadruple{\fbox{$H$}}{s_1}{ s_2}{\varphi} }
  \end{mathpar}

  This rule is naturally sound, as it forces us to cover all possible ways to transition out of $s_1$. However, instead of solving the initial problem, it makes it even worse by generating redundant proof obligations.
  Indeed, if we have a path $s_1 \to s^1_1 \to s^2_1 ... \to s'_1$, this rule will generate a proof obligation for all $s^i_1$. Instead, what we would like is to merge as many steps as possible to minimize the number of proof obligations.
  To achieve this goal, we need to be able to distinguish \textit{deterministic} sequences of computations from \textit{nondeterministic} ones. To this effect, we introduce the \textit{deterministic} transition relation $\to_\mathit{det}$ defined as follows:
  \begin{align*}
    s \to_\mathit{det} s' &\iff s \to s' \wedge (\forall e \forall s'', s \xrightarrow{e} s'' \implies s' = s'' \wedge e = \emptyset)
  \end{align*}

  \noindent Intuitively, a deterministic transition $s_1 \to_\mathit{det} s_2$ indicates that $s_2$ is the only possible successor of $s_1$ and further, that no event can be emitted by the transition to $s_2$. Using $\to_\mathit{det}$ instead of $\to$, we obtain the following sound rules:

  \begin{mathpar}
    \inferrule[Steps-L]{ s_1 \to^*_\mathit{det} s'_1 \\ \iquadruple{\fbox{$H$}}{s'_1}{s_2}{\varphi} }{ \iquadruple{\fbox{$H$}}{s_1}{s_2}{\varphi} } \qquad
    \inferrule[Steps-R]{ s_2 \to^* s'_2 \\ \iquadruple{\fbox{$H$}}{s_1}{s'_2}{\varphi} }{ \iquadruple{\fbox{$H$}}{s_1}{s_2}{\varphi} }
  \end{mathpar}

  \subsection{Soundness Proof}
  \label{sec:fei-proof}

  The initial claim that motivated the previous section is that $\fei_\varphi$ is a sound approximation of the hyperproperty $\forall\exists\LTLsquare\varphi$.
  This section is dedicated to the proof of this fact. While the high-level intuition is relatively straightforward, the formal proof turns out to be more involved. We start by giving an intuitive informal proof, and we then present a formal proof via the axiom of functional choice. The latter has been mechanized in the Coq proof assistant.

  \subsubsection*{Informal proof}

  Let $\varphi \subseteq \mathcal{E}_1 \times \mathcal{E}_2$ be a relation on events, $s_1 \in \mathcal{S}_1$ and $s_2 \in \mathcal{S}_2$ be two states, and assume that $(s_1, s_2) \in \fei_\varphi$. Under these assumptions, we would like to prove \[
    \forall \tau_1 \in \mathit{Traces}(s_1), \exists \tau_2 \in \mathit{Traces}(s_2), (\tau_1, \tau_2) \in \LTLsquare \varphi
  \]

  Given a trace $\tau_1 \in \mathit{Traces}(s_1)$, we need to construct an appropriate trace $\tau_2$ to instantiate the existential quantifier.
  Since $\tau_1 \in \mathit{Traces}(s_1)$, there has to exist some state $s'_1$ such that \[
    (1)~~\iosteps{s_1}{\tau_1[0]}{s'_1} \qquad (2)~~\tau_1[1...] \in \mathit{Traces}(s'_1)
  \]
  Since $(s_1, s_2) \in \fei_\varphi = \nu . \feif_\varphi$, by unfolding the fixed point equation $\fei_\varphi = \feif_\varphi(\nu . \feif_\varphi)$, unfolding the definition of $\feif$, and using (1) we obtain the existence of an event $e^0_2$ and a state $s'_2$ such that \[
    (3)~~\iosteps{s_2}{e^0_2}{s'_2} \qquad (4)~~(\tau_1[0], e^0_2) \in \varphi \qquad (5)~~(s'_1, s'_2) \in \fei_\varphi
  \]

  We can repeat this reasoning from the assumption that $(s'_1, s'_2) \in \fei_\varphi$ and from (2) to obtain a state $s''_1$, and event $e^1_2$, and a state $s''_2$ such that \[
    (6)~~\iosteps{s'_2}{e^1_2}{s''_2} \qquad
    (7)~~(\tau_1[1], e^1_2) \in \varphi \qquad
    (8)~~(s''_1, s''_2) \in \fei_\varphi
  \]
  
  Iterating this process an infinite number of times, we obtain a sequence of events $\tau_2 = e^0_2e^1_2e^2_2...$. By construction $\tau_2 \in \mathit{Traces}(s_2)$ and clearly $(\tau_1, \tau_2) \in \LTLsquare\varphi$ since $(\tau[i], e^i_2) \in \varphi$.

  \subsubsection*{Formal Proof via Classical Choice}

  The main challenge when mechanizing the previous proof in the Coq proof assistant is to formalize the intuition of "iterating this process an infinite number of times".
  We propose a solution in four stages: \begin{enumerate}
    \item We start by defining an abstract labeled transition system between \textit{proof states}. Transitions are labeled by events.
    \item We prove that this abstract system is progressive: for any proof state, we can always transition to a successor state.
    \item Using the axiom of functional choice and (2), we extract a deterministic and executable version of the abstract transition system that produces exactly one trace
    \item Using (3), we compute the unique trace of the abstract transition system.
    The system is constructed in such a way that this trace is guaranteed to be a valid witness for the existential quantifier.
  \end{enumerate}

  Given a relation $\varphi \subseteq \mathcal{E}_1 \times \mathcal{E}_2$, we define the set of \textit{proof states} $\Pi_\varphi$ as follows: \begin{align*}
    \Pi_\varphi &\subseteq \mathcal{S}_1 \times \mathcal{S}_2 \times \mathcal{S}^\omega_1 \triangleq \{ \ (s_1, s_2, \tau_1) \mid \tau_1 \in \mathit{Traces}(s_1) \wedge (s_1, s_2) \in \fei_\varphi \ \}\\
  \end{align*}

  We then define the following labeled \textit{proof-step} relation $\xrightarrow{-}_\varphi \subseteq \Pi_\varphi \times \mathcal{E}_2 \times \Pi_\varphi$: \[
    \inferrule{ \tau_1 = e_1 \cdot \tau'_1 \\ \iosteps{s_1}{e_1}{s'_1} \\ \iosteps{s_2}{e_2}{s'_2} \\ (e_1, e_2) \in \varphi }{(s_1, s_2, \tau_1) \xrightarrow{e_2}_\varphi
    (s'_1, s'_2, \tau'_1)}
  \]

  This transition relation describes an abstract transition system over proof states, and such that each transition generates an event in $\mathcal{E}_2$.
  Importantly, it can be proven that this abstract transition system is progressive: every proof state has a successor.

  \begin{lemma}[Progress]
    $\forall \pi \in \Pi_\varphi, \exists (e_2, \pi') \in \mathcal{E}_2 \times \Pi_\varphi, \pi \xrightarrow{e_2}_\varphi \pi'$
  \end{lemma}
  \begin{proof}
    Let $(s_1, s_2, \tau_1) \in \Pi_\varphi$. By definition we have $\tau_1 \in \mathit{Traces}(s_1)$ and $(s_1, s_2) \in \fei_\varphi$.
    By unfolding the the definition of $\mathit{Traces}$, we easily obtain the existence of a state $s'_1$ such that \[
      \iosteps{s_1}{\tau_1[0]}{s'_1} \wedge \tau_1[1...] \in \mathit{Traces}(s'_1)
    \]
    By unfolding the definition of $\fei_\varphi$ and using the fact that $\iosteps{s_1}{\tau_1[0]}{s'_1}$, we obtain the existence of an event $e_2$ and a state $s'_2$ such that \[
        \iosteps{s_2}{e_2}{s'_2} \wedge (\tau_1[0], e_2) \in \varphi \wedge (s'_1, s'_2) \in \fei_\varphi
    \]
    Clearly, we have $(s'_1, s'_2, \tau[1...]) \in \Pi_\varphi$, and $(s_1, s_2, \tau_1) \xrightarrow{e_2}_\varphi (s'_1, s'_2, \tau_1[1...])$, which concludes our proof.
  \end{proof}

  By the axiom of functional choice, the previous lemma guarantees the existence of a function $\texttt{progress} : \Pi_\varphi \to \mathcal{E}_2 \times \Pi_\varphi$ that computes an event and a successor state. We start by recalling the axiom in set-theoretic terms for clarity (the Coq implementation of the proof relies on a type-theoretic formulation of the same axiom).

  \begin{definition}[Functional Choice]
    Let $A$ and $B$ be two sets and $R \subseteq A \times B$ a binary relation such that $\forall a \in A, \exists b \in B, (a, b) \in R$. Then, there exists a function $f : A \to B$ such that $\forall a \in A, (a, f(a)) \in R$
  \end{definition}

  \begin{corollary}
    There exists a function $\texttt{progress}_\varphi : \Pi_\varphi \to \mathcal{E}_2 \times \Pi_\varphi$ such that for any proof state $\pi$ we have \[
      (e, \pi') = \texttt{progress}_\varphi(\pi) \implies \pi \xrightarrow{e}_\varphi \pi'
    \]
  \end{corollary}

  Using the \texttt{progress} function provided by this corollary, for any proof state $\pi$ we can construct a sequence $\texttt{witness}(\pi) \in \mathcal{E}_2$ defined coinductively as follows: \begin{align*}
    \texttt{witness}(\pi) &\triangleq \textbf{let} \ (e, \pi') = \texttt{progress}(\pi) \ \textbf{in} \ e \cdot \texttt{witness}(\pi')
  \end{align*}

  By construction, the $\texttt{witness}$ function satisfies the following properties.

  \begin{lemma}[Witness]
    Let $\pi = (s_1, s_2, \tau_1) \in \Pi_\varphi$ and let $\tau_2 = \texttt{witness}(\pi)$. Then $\tau_2 \in \mathit{Traces}(s_2)$ and $(\tau_1, \tau_2) \in \LTLsquare\varphi$
  \end{lemma}
  \begin{proof}
    Both facts follow from the definition of the proof-step relation $\to_\varphi$.
    Indeed, if we have a sequence of proof states \[
      (s^0_1, s^0_2, \tau^0_1) \xrightarrow{e^0_2}_\varphi (s^1_1, s^1_2, \tau^1_1) \xrightarrow{e^1_2}_\varphi (s^2_1, s^2_2, \tau^2_1) \ \ldots
    \] then in particular $\forall i \in \mathbb{N}, \iosteps{s^i_2}{e^i_2}{s^{i + 1}_2}$ (thus proving that $\tau_2$ is a valid trace), and $\forall i \in \mathbb{N}, (\tau^i_1[0], e^i_2) \in \varphi$ (thus proving that $(\tau_1, \tau_2) \in \LTLsquare\varphi$).
  \end{proof}

  We can now easily prove the soundness of $\fei_\varphi$ (\Cref{thm:fei-sound}).
  \begin{proof}
    Suppose $(s_1, s_2) \in \fei_\varphi$, and let $\tau_1 \in \mathit{Traces}(s_1)$, we show \[
      \exists \tau_2 \in \mathit{Traces}(s_2), (\tau_1, \tau_2) \in \LTLsquare \varphi
    \]
    We pose $\pi \triangleq (s_1, s_2, \tau_1)$. Clearly, $\pi \in \Pi_\varphi$. We then choose $\tau_2 = \texttt{witness}(\pi)$. By the previous lemma, $\tau_2 \in \mathit{Traces}(s_2)$ and $(\tau_1, \tau_2) \in \LTLsquare\varphi$.
    If we consider $s_1 = \mathcal{I}_1$ and $s_2 = \mathcal{I}_2$, the above construction gives a proof of $\models \forall^{\mathit{TS}_1}\exists^{\mathit{TS}_2} \LTLsquare \varphi$.
  \end{proof}
  \section{From Relational Invariants to Safety Relations}
  \label{sec:fe}

  \newcommand{\fef}{\texttt{feF}}
  \newcommand{\fe}{\texttt{fe}}

  While formulas of the form $\forall\exists\LTLsquare\varphi$ are already covering relevant hyperproperties such as trace-inclusion and several variants of generalized non-interference,
  many temporal hyperproperties of interest are of the more general form $\forall\exists\psi$ where $\psi$ is an arbitrary \textit{safety relation} between traces.
  In this section, we extend the coinductive relation $\fei$ to support 
  reasoning about arbitrary safety relations. Further, we develop reasoning rules
  for the specific case where safety relations are defined using LTL-style temporal modalities.

  \subsection{It's All About Staying Safe}
  
  Intuitively, safety relations are relations that specify the absence of bad interactions between two traces. More formally, safety relations between infinite traces are exactly these for which it suffices to look at finite prefixes to determine whether the relation is violated. In other words, for every safety relation there exists a so called \textit{monitor} that, given a pair of trace prefixes, checks wether or not they are compatible with the relation. Further, two traces satisfy a safety relation if and only if all their prefixes are accepted by the monitor.
  Following this intuition, establishing a safety property $\psi$ can always be reduced to a proof of invariance: it suffices to run the system and the monitor for $\psi$ concurrently, and to prove the invariant "the monitor never reports any violation".
  However, extracting an appropriate monitor for any arbitrary safety relation $\psi$ is non-trivial, especially when the relation $\psi$ is expressed in an expressive temporal logic. One way to systematically extract a monitor from a high-level logical specification
  is by computing a finite-word \textit{safety automaton} that accepts all pairs of prefixes
  compatible with a trace relation.
  While this approach works well for automated verification (because verification can then be reduced to efficient operations on automata), it it not a satisfactory solution for interactive verification inside a proof assistant. First of all, translating a high-level specification expressed in a temporal logic to an equivalent automaton is a non-trivial task. Further, we would also need to formally verify the correctness of this translation.
  Finally, even assuming that we have a certified algorithm to translate high-level specifications to automata \cite{10.1007/978-3-642-39799-8_31}, we would then need to reason about states and transitions of these automata rather than reasoning about the original logical specifications. This would lead to tedious and unintuitive proofs.
  Instead, we propose a solution via the notion of \textit{property derivatives} \cite{derivatives, 10.1007/978-3-319-15579-1_22}.
  The key insight is to observe that given a safety relation $\psi$ and two traces $\tau_1, \tau_2$, there is a systematic way to decompose a proof of $(\tau_1, \tau_2) \in \psi$ into a condition on $\tau_1[0], \tau_2[0]$ expressing \textit{what needs to hold now} in order for $\psi$ to not be immediately violated; and 
  a proof of $(\tau_1[1...], \tau_2[1...]) \in \psi'$ for some carefully chosen $\psi'$ expressing what need to hold \textit{next} in order for $\psi$ to not be violated later.
  In this section, we exploit this observation to extend the coinductive relation $\fei$ presented in the previous section to a more general relation $\fe$ that supports reasoning about arbitrary safety relations.

  \subsection{Coinductive Safety Proofs}

  Formally, safety relations between traces can be defined in terms of a \textit{safety closure} operator $|-|_\mathit{safe} : \functor{\tracerel{\mathcal{E}_1}{\mathcal{E}_2}}$.
  Given a trace relation $\psi$, the safety closure $|\psi|_\mathit{safe}$ is extracting the \textit{safety} part of $\psi$ by relating all pairs of traces such that all their prefixes are \textit{compatible} with $\psi$ (i.e., there exists a way to extend the prefixes into infinite traces that satisfy $\psi$).
  
  \begin{definition}[Safety closure]
    Let $\psi \subseteq \mathcal{E}^\omega_1 \times \mathcal{E}^\omega_2$ be a binary trace relation.
    The safety closure of $\psi$, noted $|\psi|_\mathit{safe}$ is the binary trace relation defined as follows: \[
      |\psi|_\mathit{safe} = \{ (\tau_1, \tau_2) \mid \forall n, \exists \tau'_1, \tau'_2 \in \mathcal{E}^\omega_1 \times \mathcal{E}^\omega_2, (\tau_1[n...] \cdot \tau'_1, \tau_2[n..] \cdot \tau'_2) \in \psi \}
    \]
  \end{definition}

  Safety properties are then defined to be all fixed points of $|-|_\mathit{safe}$, i.e., all the trace relations that are equivalent to their safety closure. Since $\psi \subseteq |\psi|_\mathit{safe}$ for any $\psi$, safety relations can be equivalently characterized as prefixed points of $|-|_\mathit{safe}$.

  \begin{definition}[Safety relations]
    A trace relation $\psi$ is said to be a \textit{safety} relation iff $|\psi|_\mathit{safe} \subseteq \psi$
  \end{definition}

  We observe that the safety closure of a relation can equivalently be defined as a greatest fixed point. We start by introducing an operator $\Delta$, inspired by the notion of \textit{derivatives} introduced by Brzozowski \cite{derivatives}. Given two events $e_1 \in \mathcal{E}_1, e_2 \in \mathcal{E}_2$, the derivative of $\psi$, noted $\Delta_{e_1, e_2}(\psi)$, is the relation that should be satisfied by suffixes $\tau_1 \in \mathcal{E}^\omega_1, \tau_2 \in \mathcal{E}^\omega_2$ in order for the traces $(e_1\tau_1, e_2\tau_2)$ to satisfy $\psi$. Formally, we have: \[
      \Delta_{e_1, e_2}(\psi) \triangleq \{ \ (\tau_1, \tau_2) \mid (e_1\tau_1, e_2\tau_2) \in \psi \ \}
  \]

  Two traces $\tau_1 = e_1\tau'_1$ and $\tau_2 = e_2t'_2$ are in the safety closure of $\varphi$ if the derivative $\Delta_{e_1, e_2}(\psi)$ is non-empty (i.e., the first two symbols are not trivially violating the relation), and if the suffixes $\tau'_1, \tau'_2$ are again in the closure of $\Delta_{e_1, e_2}(\psi)$.
  This intuition is formally captured by the definition of a ternary coinductive relation
  $\texttt{safe} \in \mathcal{P}(\mathcal{E}^\omega_1 \times \mathcal{E}^\omega_2 \times \tracerel{\mathcal{E}_1}{\mathcal{E}_2})$ \begin{align*}
    \texttt{safeF}(C) &\triangleq \{ (e_1\tau_1, e_2\tau_2, \psi) \mid \Delta_{e_1, e_2}(\varphi) \ne \emptyset \wedge (\tau_1, \tau_2, \Delta_{e_1, e_2}(\psi)) \in C \}\\
    \texttt{safe} &\triangleq \nu . \texttt{safeF}
  \end{align*}
  The relation $\texttt{safe}$ precisely characterizes the safety closure in the sense of the following lemma:

  \begin{lemma}
    \label{thm:cosafe}
    $(\tau_1, \tau_2) \in |\psi|_\mathit{safe} \iff (\tau_1, \tau_2, \psi) \in \texttt{safe}$
  \end{lemma}
  \begin{proof} \
    \begin{itemize}
      \item $(\Longleftarrow)$ By induction on the length of the prefixes in the definition of $|-|_\mathit{safe}$, unfolding the fixed point equation $\texttt{safe} = \texttt{safeF}(\texttt{safe})$, and
      exploiting the fact that $\Delta_{e_1, e_2}(\psi) \ne \emptyset \iff \exists (\tau_1, \tau_2), (e_1\tau_1, e_2\tau_2) \in \psi \ (*)$.
      \item $(\Longrightarrow)$ By coinduction with $R = \{ \ (\tau_1, \tau_2, \psi) \mid (\tau_1, \tau_2) \in \psi \ \}$ as the coinduction hypothesis.
      The fact that $R$ is a postfixed point of $\texttt{safeF}$ follows from $(*)$ and the additional observation that $(e_1\tau_1, e_2\tau_2) \in |\psi|_\mathit{safe} \implies (\tau_1, \tau_2) \in |\Delta_{e_1, e_2}(\psi)|_\mathit{safe}$.
    \end{itemize}
  \end{proof}

  As an immediate corollary, we can use coinduction to prove that a pair of traces satisfies any safety relation.

  \begin{corollary}[Coinductive safety proofs]
    Let $\psi$ be a safety property, then \[
      (\tau_1, \tau_2) \in \psi \iff (\tau_1, \tau_2) \in |\psi|_\mathit{safe} \iff (\tau_1, \tau_2, \psi) \in \texttt{safe}
    \]
  \end{corollary}

  A useful interpretation of this corollary is that, given two traces $\tau_1, \tau_2$ and a safety property $\psi$, the derivative operator gives a systematic way to split the proof of $(\tau_1, \tau_2) \in \psi$ into \textit{what needs to hold now} (i.e., $\Delta_{\tau_1[0], \tau_2[0]}(\psi) \ne \emptyset$) and \textit{what needs to hold next} (i.e., $(\tau_1[1...], \tau_2[1...], \Delta_{\tau_1[0], \tau_2[0]}(\psi)) \in \texttt{safe}$). The next subsection builds on this intuition to construct a coinductive relation $\fe$ that extends $\fei$ with support for reasoning about arbitrary safety relations.

  \subsection{A Coinductive Relation Supporting Safety Specifications}

  For the rest of this section, we fix two LTSs $\mathit{TS}_1 = (\mathcal{S}_1, \mathcal{E}_1, \mathcal{I}_1, \to_1)$ and $\mathit{TS}_2 = (\mathcal{S}_2, \mathcal{E}_2, \mathcal{I}_2, \to_2)$.
  To extend $\fei$ to the more general case of arbitrary safety relations, we replace the functor $\feif_{\varphi} : \mathcal{P}(\mathcal{S}_1 \times \mathcal{S}_2) \xrightarrow{\mathit{mon}} \mathcal{P}(\mathcal{S}_1 \times \mathcal{S}_2)$, with a new functor $\fef$ that has the following signature: \[
    \fef : \functor{\mathcal{P}(\mathcal{S}_1 \times \mathcal{S}_2 \times \mathcal{P}(\mathcal{E}^\omega_1 \times \mathcal{E}^\omega_2))}
  \]
  The associated greatest fixed point $\fe \triangleq \nu . \fef$ relates two states and a trace relation. The intention is that for any safety relation $\psi$, a proof of $(s_1, s_2, \psi) \in \fe$ should guarantee that $(s_1, s_2) \in \forall^{\mathit{TS}_1}\exists^{\mathit{TS}_2}\psi$.
  Importantly, with this more general signature, the trace relation $\psi$ can change over the course of a coinductive proof.
  This will allow us to use the notion of derivatives to \textit{unroll} $\psi$.
  We start by defining an operator $\texttt{next}^R_{e_1, e_2}(\psi) \subseteq \mathcal{S}_1 \times \mathcal{S}_2$: \[
    \texttt{next}^R_{e_1, e_2}(\psi) \triangleq \{ \ (s_1, s_2) \mid \Delta_{e_1, e_2}(\psi) \ne \emptyset \wedge (s_1, s_2, \Delta_{e_1, e_2}(\psi)) \in R \ \}
  \]
  We use \texttt{next} to define the coinductive relation $\fe$: \begin{align*}
    \fef(R) &\triangleq \{ \ (s_1, s_2, \psi) \mid
   \forall \iosteps{s_1}{e_1}{s'_1}, \exists \iosteps{s_2}{e_2}{s'_2}, (s'_1, s'_2) \in \texttt{next}^R_{e_1, e_2}(\psi) \ \}\\
    \fe &\triangleq \nu . \fef
  \end{align*}

  \begin{theorem}[Soundness of $\fe$]
    \label{thm:fe}
    Let $\psi$ be a safety relation, \[
      (\mathcal{I}_1, \mathcal{I}_2, \psi) \in \fe \implies \models \forall^{\mathit{TS}_1}\exists^{\mathit{TS}_2}\psi
    \]
  \end{theorem}
  \begin{proof}
    The proof follows exactly the same structure as for $\fei$, we just slightly change the abstract transition system. We extend the set of proof states to be $\Pi \triangleq \{ \ (s_1, s_2, \tau_1, \psi) \mid \tau_1 \in \mathit{Traces}(s_1) \wedge (s_1, s_2, \psi) \in \fe \ \}$ and the transition relation is defined as follows: \begin{mathpar}
      \inferrule{ \tau_1 = e_1\tau'_1 \\ \iosteps{s_1}{e_1}{s'_1} \\ \iosteps{s_2}{e_2}{s'_2} \\ \Delta_{e_1, e_2}(\psi) \ne \emptyset }{ (s_1, s_2, \tau_1, \psi) \xrightarrow{e_2}_\mathit{proof} (s'_1, s'_2, \tau'_1, \Delta_{e_1, e_2}(\psi))}
    \end{mathpar}
    Given two states and a safety relation $\psi$ such that $(s_1, s_2, \psi) \in \fe$,
    and provided $\tau_1 \in \mathit{Traces}$, we can execute the abstract transition system from the initial proof state $(s_1, s_2, \tau_1)$ to obtain a trace $\tau_2$.
    We pick $\tau_2$ as a witness for the existential quantifier. It remains to prove that $(\tau_1, \tau_2) \in \psi$. Since $\psi$ is a safety relation, by \ref{thm:cosafe} it suffices to show $(\tau_1, \tau_2, \psi) \in \texttt{safe}$.
    This fact is easily proven by coinduction (exploiting $(s_1, s_2, \psi) \in \fe$ and the fact that $\tau_2$ is generated by executing the abstract transition system defined by $\to_\mathit{proof}$).
  \end{proof}

  As for $\fei$, the definition of $\fe$ as a greatest fixed point immediately gives us the ability to perform incremental proofs by parameterized coinduction.
  To use $\fe$ in incremental proofs, we redefine our semantic quadruples using $\fef$ instead of $\feif$ as the underlying functor: \begin{align*}
    \quadruple{\fbox{$H$}}{s_1}{s_2}{\psi} &\triangleq (s_1, s_2, \psi) \in G_{\fef}(H)\\
    \quadruple{\dbox{$H$}}{s_1}{s_2}{\psi} &\triangleq (s_1, s_2, \psi) \in H \cup G_{\fef}(H)
  \end{align*}

  Note that here, $\psi$ is not just a relation between events, but a relation between infinite traces.
  For convenience, we also introduce the following sextuple: \[
    \quadruple{\fbox{$H$}}{e_1 \triangleright s_1}{e_2 \triangleright s_2}{\psi} \triangleq (s_1, s_2) \in \texttt{next}^{H \ \cup \ G_{\fef}(H)}_{e_1, e_2}(\psi)
  \]

  The corresponding reasoning rules are presented in \Cref{fig:fe_core}.

  \begin{figure}[h!]
      \begin{mathpar}
        \inferrule[Init]{ \psi \ \textrm{is a safety relation} \\ \quadruple{\ \fbox{$\emptyset$}}{\mathcal{I}_1}{\mathcal{I}_2}{\psi}}{ \models \forall^{\mathit{TS}_1}\exists^{\mathit{TS}_2}\psi } \qquad
        \inferrule[Cycle]{ (s_1, s_2, \psi) \in H }{ \quadruple{\ \dbox{$H$}}{s_1}{s_2}{\psi} }\\
        \inferrule[Step]{\forall \iosteps{s_1}{e_1}{s'_1}, \exists \iosteps{s_2}{e_2}{s'_2}, \quadruple{\ \fbox{$H$}}{ e_1 \triangleright s_1'}{ e_2 \triangleright s_2'}{\psi} }{ \quadruple{\ H^{?}}{s_1}{s_2}{\psi} }\\
        \inferrule[Invariant]{ (s_1, s_2, \psi) \in H' \\ \forall (s_1, s_2, \psi) \in H', \quadruple{\ \fbox{$H \cup H'$}}{s_1}{s_2}{\psi} }{ \quadruple{\ H^{?}}{s_1}{s_2}{\psi} }\\
        \inferrule[Deriv]{ \Delta_{e_1, e_2}(\psi) \ne \emptyset \\ \quadruple{\ \dbox{$H$}}{s_1}{s_2}{\Delta_{e_1, e_2}(\psi)}}{ \quadruple{\ \fbox{$H$}}{e_1 \triangleright s_1}{e_2 \triangleright s_2}{\psi} }
      \end{mathpar}
      \caption{Incremental reasoning rules for $\fe$}
      \label{fig:fe_core}
      \Description{Incremental reasoning rules for $\fe$}
  \end{figure}
  
  The rules \textsc{Init}, \textsc{Cycle}, \textsc{Step}, and \textsc{Invariant} are similar to the ones of \Cref{fig:fei_core}. The main difference is that the \textsc{Step} rule \textit{suspends} the proof temporarily by requiring us to check whether the choice of transition $\iosteps{s_2}{e_2}{s'_2}$ in reaction to a transition $\iosteps{s_1}{e_1}{s'_1}$ is \emph{compatible} with the safety relation $\psi$. This is reflected by a premise of the form $\quadruple{\fbox{$H$}}{e_1 \triangleright s_1}{e_2 \triangleright s_2}{\psi}$ that can be discharged using the rule \textsc{Deriv}.
  The rule \textsc{Deriv} requires us to prove two things:
  \begin{enumerate}
    \item $\psi$ should not be immediately violated by the choice of event $e_2$; formalized by $\Delta_{e_1, e_2}(\psi) \ne \emptyset$
    \item $\psi$ should not be violated later; formalized by $\quadruple{\dbox{$H$}}{s_1}{s_2}{\Delta_{e_1, e_2}(\psi)}$
  \end{enumerate}
  Note that the guard around $H$ is only released after an application of \textsc{Deriv}.
  In addition to the rules presented in \Cref{fig:fe_core}, $\fe$ also inherits the \textsc{Steps-L} and \textsc{Steps-R} rules of $\fei$ (adapted to triples $(s_1, s_2, \psi)$ in an obvious way).

  \subsection{Handling Derivatives}

  While all rules presented in \Cref{fig:fe_core} are sound, they can be difficult to use in practice because of the \textit{semantic} treatment of derivatives (as opposed to a more \textit{syntactic} one). First, the \textsc{Deriv} rule requires us to prove the non-emptiness of a trace relation. Further, after a relation $\psi$ has been derived, the proof focuses on a new relation $\Delta_{e_1, e_2}(\psi)$.
  There is \textit{a priori} no reason to believe that $\Delta_{e_1, e_2}(\psi)$ has been encountered before, and therefore, it is not obvious that it will later be possible to use the rule \textsc{Cycle} to conclude the proof.
  Fortunately, when $\psi$ is expressed using LTL modalities, there is a systematic way to determine an LTL representation of $\Delta_{e_1, e_2}(\psi)$ by recursing on the syntactic structure of the formula. \Cref{fig:derivatives} describes how\footnote{To the best of our knowledge, the notion of derivatives has never explicitly been formalized for LTL-defined languages. However, derivatives have been applied to languages over infinite words expressed as omega-regular expressions \cite{10.1007/978-3-319-15579-1_22}. Further, derivatives are easily recovered from the expansion laws of temporal modalities \cite{Baier2008PrinciplesOM}. }.

  \begin{figure}[h!]
    \begin{align*}
      \Delta_{e_1, e_2}(\LTLsquare \psi) &\equiv \Delta_{e_1, e_2}(\psi) \wedge \LTLsquare\psi\\
      \Delta_{e_1, e_2}(\psi_1 \wedge \psi_2) &\equiv \Delta_{e_1, e_2}(\psi_1) \wedge \Delta_{e_1, e_2}(\psi_2)\\
      \Delta_{e_1, e_2}(\psi_1 \vee \psi_2) &\equiv \Delta_{e_1, e_2}(\psi_1) \vee \Delta_{e_1, e_2}(\psi_2)\\
      \Delta_{e_1, e_2}(\psi_1 \ \texttt{W} \ \psi_2) &\equiv \Delta_{e_1, e_2}(\psi_2) \vee [\Delta_{e_1, e_2}(\psi_1) \wedge (\psi_1 \ \texttt{W} \ \psi_2)]\\
      \Delta_{e_1, e_2}(\bigcirc \psi) &\equiv \psi\\
      \Delta_{e_1, e_2}(\varphi) &\equiv \begin{cases}
        \begin{aligned}
          ~&\mathit{true} &&\quad\textrm{if $(e_1, e_1) \in \varphi$}\\
          ~&\mathit{false} &&\quad\textrm{otherwise}
        \end{aligned}
      \end{cases}
    \end{align*}
    \caption{Derivatives of common temporal modalities}
    \label{fig:derivatives}
    \Description{Derivatives of common temporal modalities}
  \end{figure}

  We note that for temporal modalities $\LTLsquare$ and $\texttt{W}$, the original formula occurs as a subformula in the derivative.
  In a sense, the $\Delta$ operator \textit{reveals} the coinductive nature of these modalities. This will help us closing cycles in proofs.

  \begin{example}
    We consider the following transition systems:

  \begin{center}
    \begin{tikzpicture}
      \node (X) {$\mathit{TS}_1$:};
      \node[initial, state, below right = of X, yshift=.5cm] (A) {$q_0$};
      \draw[->] (A) edge[loop right] node{$a$} (A);
    \end{tikzpicture}
    \begin{tikzpicture}
      \node (X) {$\mathit{TS}_2$:};
      \node[initial, state, below right = of X, yshift=.5cm] (A) {$s_0$};
      \node[state, right = of A] (B) {$s_1$};
      \draw[->] (A) edge node{$a$} (B);
      \draw[->] (B) edge[loop right] node{$b$} (B);
    \end{tikzpicture}
  \end{center}

  \noindent We wish to prove $\models \forall^{\mathit{TS}_1}\exists^{\mathit{TS}_2}(a_1 \wedge a_2) \ \texttt{W} \ b_2$. We start by applying \textsc{Init}, followed by \textsc{Step} to match the $a$-transition in $\mathit{TS}_1$ with an $a$-transition in $\mathit{TS}_2$. \begin{align*}
    & \models \forall^{\mathit{TS}_1}\exists^{\mathit{TS}_2}(a_1 \wedge a_2) \ \texttt{W} \ b_2\\
    \Longleftarrow \quad & \quadruple{\fbox{$\emptyset$}}{q_0}{s_0}{(a_1 \wedge a_2) \ \texttt{W} \ b_2} & \textrm{\color{gray}(by \textsc{Init})}\\
    \Longleftarrow \quad & \quadruple{\fbox{$\emptyset$}}{a \triangleright q_0}{a \triangleright s_1}{(a_1 \wedge a_2) \ \texttt{W} \ b_2} & \textrm{\color{gray}(by \textsc{Step})}
  \end{align*}
  To make further progress, we need to prove that emitting two $a$'s does not immediately violate the current temporal relation. To do so, we apply the rule \textsc{Deriv} and we obtain \[
    \Delta_{a, a}((a_1 \wedge a_2) \ \texttt{W} \ b_2) \ne \emptyset \wedge \quadruple{\dbox{$\emptyset$}}{q_0}{s_1}{\Delta_{e_1, e_2}((a_1 \wedge a_2) \ \texttt{W} \ b_2)} \quad \textrm{\color{gray}(by \textsc{Deriv})}
  \]
  Using known derivatives (\Cref{fig:derivatives}), we can easily show that $\Delta_{a, a}((a_1 \wedge a_2) \ \texttt{W} \ b_2) \equiv (a_1 \wedge a_2) \ \texttt{W} \ b_2 \ne \emptyset$. After substituting the result of the derivative we are therefore left with the goal \[
      \quadruple{\dbox{$\emptyset$}}{q_0}{s_1}{(a_1 \wedge a_2) \ \texttt{W} \ b_2}
  \]
  By \textsc{Step} we match an $a$-transition in $\mathit{TS}_1$ with a $b$-transition in $\mathit{TS}_2$ and we obtain \[
      \quadruple{\fbox{$\emptyset$}}{a \triangleright q_0}{b \triangleright
    s_1}{(a_1 \wedge a_2) \ \texttt{W} \ b_2} \quad \textrm{\color{gray}(by \textsc{Step})}
  \]
  Again, by \textsc{Deriv}, we have to prove two goals: \[
      \Delta_{a, b}((a_1 \wedge a_2) \ \texttt{W} \ b_2) \ne \emptyset \wedge 
      \quadruple{\dbox{$\emptyset$}}{q_0}{s_1}{\Delta_{e_1, e_2}((a_1 \wedge a_2) \ \texttt{W} \ b_2)} \quad \textrm{\color{gray}(by \textsc{Deriv})}
  \]
  Since $\Delta_{a, b}((a_1 \wedge a_2) \ \texttt{W} \ b_2) \equiv \mathit{true} \ne \emptyset$, it suffices to prove the following: \[
    \quadruple{\dbox{$\emptyset$}}{q_0}{s_1}{ \mathit{true} }
  \]
  Using \textsc{Invariant} we set the current coinduction hypothesis to be $H \triangleq \{ \ (q_0, s_1, \mathit{true}) \ \}$ \[
      \quadruple{\fbox{$H$}}{q_0}{s_1}{ \mathit{true} } \quad \textrm{\color{gray}(by \textsc{Invariant})}
  \]
  By \textsc{Step} we again match an $a$-transition with a $b$-transition and we get \[
      \quadruple{\fbox{$H$}}{a \triangleright q_0}{b \triangleright s_1}{\mathit{true}} \quad \textrm{\color{gray}(by \textsc{Step})}
  \]
  Since the derivative of $\mathit{true}$ is $\mathit{true}$ for any pair of events, by \textsc{Deriv} we obtain \[
      \quadruple{\dbox{$H$}}{q_0}{s_1}{\mathit{true}} \quad \textrm{\color{gray}(by \textsc{Deriv})}
  \]
  Finally, since $(q_0, s_1, \mathit{true}) \in H$, we can conclude using \textsc{Cycle}.
  \qed
  \end{example}
  
  While it is correct, the proof presented above suffers from several limitations. 
  
  \paragraph*{Limitation 1} 
  Even though we purposely gave a very detailed proof, the example discussed above still seems abnormally long.
  In particular, the last application of the rule \textsc{Step} seems redundant. Indeed, if we look at the structure of the two transition systems, the proof should only have two easy steps: one to match the $a$-transition with an $a$-transition, and a second one to match an $a$-transition with a $b$-transition. More precisely, once we reached the goal $\quadruple{\fbox{$\emptyset$}}{q_0}{s_1}{\mathit{true}}$, we would like to be able to conclude right away. This goal should, in principle, immediately follow from the previous and intuitively more difficult goal $\quadruple{\fbox{$\emptyset$}}{q_0}{s_1}{(a_1 \wedge a_2) \ \texttt{W} \ b_2 }$.

  \paragraph*{Limitation 2} Another and more subtle concern is the fact that whenever we applied \textsc{Deriv}, we had to explicitly compute the derivative of the current trace relation, and substitute the result back in the current goal. Substituting a relation with an equivalent one is not an issue when doing informal reasoning in set theory;
  however, in type-theory based proof assistants such as Coq, this is not possible in general.
  By default, we can only substitute a relation for another one if both relations are equal in the sense of propositional equality.
  One solution would be to assume the axiom of \textit{predicate extensionality} that precisely stipulates that two equivalent predicates are also equal.

  \begin{axiom}[Predicate Extensionality]
    Let $P, Q \in \mathcal{P}(A)$ be predicates over $A$, then $P \equiv Q \implies P = Q$.
  \end{axiom}

  Note however that this approach is not entirely satisfactory as it 
  would make our development depends on a second axiom in addition to functional choice.

  \subsection{Strengthening Is All We Need}

  Instead of relying on additional axioms, a more general solution to address the two limitations presented in the previous section is to prove once and for all that our coinductive relation supports the following \textit{strengthening} rule: \begin{mathpar}
    \inferrule[Strengthen]{\psi' \subseteq \psi \\ \quadruple{H^?}{s_1}{s_2}{\psi'}}{\quadruple{H^?}{s_1}{s_2}{\psi}}
  \end{mathpar}
  Note that the two question marks in \textsc{Strengthen} should be interpreted as follows: the rule can be applied both when $H$ is guarded or unguarded, but it never modifies the guard (in particular, if the hypothesis was guarded, it remains guarded).
  The addition of \textsc{Strengthen} would clearly solve the first limitation we discussed: it can factor redundant proof steps by reusing results already established for stronger trace relations. Further, it also solves the substitution problem. Indeed, substituting a trace relation with an equivalent one is just a special case of strengthening.

  Unfortunately, with our current definition of semantic quadruples, \textsc{Strengthen} is unsound. To understand why, let us first observe that the following weaker version of the rule \textit{is} sound: \begin{mathpar}
    \inferrule{\psi' \subseteq \psi \\ \quadruple{\fbox{$\emptyset$}}{s_1}{s_2}{\psi'}}{\quadruple{H^?}{s_1}{s_2}{\psi}}
  \end{mathpar}
  Note that in this variant, $H$ is replaced by $\emptyset$ after applying the rule!
  In other words, if we need to prove $\quadruple{H^?}{s_1}{s_2}{\psi}$, it suffices to prove that the property holds for a stronger trace relation $\psi' \subseteq \psi$, but at the cost of forgetting all knowledge previously accumulated in the proof.
  This essentially ruins the benefits of having an incremental proof system.
  Intuitively the reason why we need to reset $H$ is that the triples we explore during the proof of $\quadruple{\fbox{$H$}}{s_1}{s_2}{\psi'}$ are not necessarily the same as the triples explored during a proof of $\quadruple{\fbox{$H$}}{s_1}{s_2}{\psi}$.
  In particular, instead of considering the successive derivatives of $\psi$, we are instead considering the successive derivatives of $\psi'$. In turn, if we use $H$ to reason about triples emanating from $(s_1, s_2, \psi')$, there is no reason to believe that the same reasoning carries over to triples emanating from $(s_1, s_2, \psi)$ without further assumptions.
  Fortunately, we can slightly generalize our coinductive relation in such a way that \textsc{Strengthen} (in its more general version) is provably sound, and all other proof rules remain valid as well. The next section discusses this generalization.

  For now, let us assume that \textsc{Strengthen} is sound. Using the equivalences listed in \Cref{fig:derivatives}, \textsc{Strengthen} can be combined with \textsc{Deriv} to obtain a set of rules for handling temporal modalities.
  \Cref{fig:deriv-hyco} presents a selection of such rules.

  \begin{figure}[h!]
      \begin{mathpar}
        \inferrule[Deriv-$\LTLsquare$]{(e_1, e_2) \in \varphi \\ \quadruple{\ \dbox{$H$}}{s_1}{s_2}{\LTLsquare\varphi}}{\quadruple{\fbox{$H$}}{e_1 \triangleright s_1}{e_2 \triangleright s_2}{\LTLsquare\varphi}}\qquad
        \inferrule[Deriv-$\LTLnext$]{\psi \ne \emptyset \\ \quadruple{\ \dbox{$H$}}{s_1}{s_2}{\psi}}{\quadruple{\fbox{$H$}}{e_1 \triangleright s_1}{e_2 \triangleright s_2}{\LTLnext\psi}}\\
        \inferrule[Deriv-\texttt{W}-Now]{(e_1, e_2) \in \varphi_2 \\ \quadruple{\ \dbox{$H$}}{s_1}{s_2}{\mathit{true}}}{\quadruple{\fbox{$H$}}{e_1 \triangleright s_1}{e_2 \triangleright s_2}{\varphi_1 \ \texttt{W} \ \varphi_2}} \qquad
        \inferrule[Deriv-\texttt{W}-Later]{(e_1, e_2) \in \varphi_1 \\ \quadruple{\ \dbox{$H$}}{s_1}{s_2}{\varphi_1 \ \texttt{W} \ \varphi_2}}{\quadruple{\fbox{$H$}}{e_1 \triangleright s_1}{e_2 \triangleright s_2}{\varphi_1 \ \texttt{W} \ \varphi_2}}
      \end{mathpar}
      \caption{Selection of proof rules for temporal modalities}
      \label{fig:deriv-hyco}
      \Description{Selection of proof rules for temporal modalities}
  \end{figure}

  \section{Up-to Techniques}

  \subsection{Reasoning Up-to Stronger Trace Relations}

  Often times, we have a goal of the form $\quadruple{\dbox{$H$}}{s_1}{s_2}{\psi}$ where the triple $(s_1, s_2, \psi)$ \textit{almost} belongs to $H$ (which would allow us to conclude using the rule \textsc{Cycle}), but not \textit{exactly}. In other words, we would like to be able to reason \textit{up to} slight enlargements of the coinduction hypothesis \cite{Pous_Sangiorgi_2011}.
  As discussed in the previous section, a typical occurrence of this pattern is when we have a goal of the form $\quadruple{\dbox{$H$}}{s_1}{s_2}{\psi}$ but $H$ only contains $(s_1, s_2, \psi')$ for some stronger relation $\psi' \subset \psi$.
  In \cite{up_to} and \cite{all_the_way_up}, Pous developed systematic ways to soundly integrate up-to techniques in coinductive proofs. Building on these ideas,
  Zakowski et al. later showed that the parameterized greatest fixed point operator $G$ can be generalized to support up-to techniques as well as other enhancements \cite{gpaco}.
  This generalization, usually referred to as GPACO, preserves the ability to perform incremental proofs. The reasoning rules of GPACO are presented in \Cref{fig:gpaco}. We note that the original definition by Zakowski et al. is more involved, and we refer the readers to \cite{gpaco} for a more in-depth presentation. However, our Coq development relies on an instance of GPACO which exposes exactly the simplified interface described in \Cref{fig:gpaco}.
  Contrary to standard parameterized coinduction, GPACO explicitly marks when the current assumption can be used or not: we note $\fbox{$G_F(H)$}$ when the assumption $H$ is guarded, and $\dbox{$G_F(H)$}$ when it is unguarded. 
  When proof rules can be applied regardless of whether the coinduction hypothesis is guarded or not, we note $G^?_F(H)$.

  \begin{figure}[ht!]
      \begin{mathpar}
        \inferrule[Init]{ X \subseteq \fbox{$G_F(\emptyset)$} }{X \subseteq \nu . F} \qquad
        \inferrule[Step]{ X \subseteq F(\dbox{$G_F(H)$})}{X \subseteq G^?_F(H)} \qquad
        \inferrule[Accumulate]{X \subseteq \fbox{$G_F(X \cup H)$}}{X \subseteq G^?_F(H)}\\
        \inferrule[Cycle]{X \subseteq H}{X \subseteq \dbox{$G_F(H)$}}\qquad
        \inferrule[Up-to]{ \mathit{clo} \ \textrm{is \emph{compatible} with $F$}}{\mathit{clo}(G^?_F(H)) \subseteq G^?_F(H)}
      \end{mathpar}
      \caption{Rules of GPACO}
      \label{fig:gpaco}
      \Description{Rules of GPACO}
  \end{figure}

  As for parameterized coinduction (\Cref{fig:coinduction}), the rule \textsc{Step} releases the guard. When the current hypothesis is unguarded, the rule \textsc{Cycle} can be invoked to conclude a proof.
  Finally, the rule \textsc{Up-to} enables up-to reasoning. The idea is that given an operator $\mathit{clo}$ that satisfies a certain compatibility criterion (Formally, $\mathit{clo}$ is \textit{compatible} with $F$ iff $\mathit{clo} \circ F \subseteq F \circ \mathit{clo}$~\cite{all_the_way_up}), the goal $X \subseteq G^?(H)$ can be replaced by the hopefully simpler goal $X \subseteq \mathit{clo}(G^?(H))$. Typically, we choose $\mathit{clo}$ to be a closure operator (hence the naming convention $\mathit{clo}$) so that $\mathit{clo}(G^?(H))$ is larger than $G^?(H)$.
  To support up-to techniques, we redefine our semantic tuples using GPACO as follows: \begin{align*}
    \quadruple{\fbox{$H$}}{s_1}{s_2}{\psi} &\triangleq (s_1, s_2, \psi) \in \fbox{$G_\fef(H)$}\\
    \quadruple{\dbox{$H$}}{s_1}{s_2}{\psi} &\triangleq (s_1, s_2, \psi) \in \dbox{$G_\fef(H)$}\\
    \quadruple{\fbox{$H$}}{e_1 \triangleright s_1}{e_2 \triangleright s_2}{\psi} &\triangleq (s_1, s_2) \in \texttt{next}^{\dbox{$G_\fef(H)$}}_{e_1, e_2}(\psi)
  \end{align*}

  All the reasoning rules previously introduced are preserved by this change. However, this generalization naturally inherits the \textsc{Up-to} rule. In particular, it allows us to reason up to stronger trace relations. The associated closure operator \texttt{strclo} is defined as follows: \begin{align*}
    \texttt{strclo}(H) \triangleq \{ \ (s_1, s_2, \psi) \mid \exists \psi' \subseteq \psi, (s_1, s_2, \psi') \in H \ \}
  \end{align*}

  By definition, to prove that a triple $(s_1, s_2, \psi)$ belongs to $\texttt{strclo}(H)$, it suffices to prove that $(s_1, s_2, \psi')$ is in $H$ for some relation $\psi' \subseteq \psi$.
  Importantly, we observe that $\texttt{strclo}$ is compatible with the monotone functor $\fef$. Consequently, \textsc{Strengthen} is just an instance of \textsc{Up-to} with $\fef$ as the underlying functor, and $\texttt{strclo}$ as the (compatible) closure operator.

  \subsection{Reasoning Up-to Simulation}

  To prove a universal trace property $\forall^\mathit{TS}\psi$, it suffices to show that the property hold for any superset $T \supseteq \mathit{Traces}(\mathit{TS})$. Dually, to prove an existential trace property $\exists^\mathit{TS}\psi$, it suffices to show that the property holds for a subset $T \subseteq \mathit{Traces}(\mathit{TS})$.
  In turn, to prove a $\forall\exists$ hyperproperty, it suffices to consider an over-approximation of the universally quantified traces and an under-approximation of the existentially quantified traces. In some cases, reasoning about approximated sets of traces can greatly simplify proofs.
  This intuition about traces can be lifted at the level of states if we consider simulation instead of trace inclusion.
  In particular, our coinductive relation can be equipped with the following sound rules: \begin{mathpar}
    \inferrule[Sim-L]{ (s_1, s'_1) \in \texttt{sim} \\ \quadruple{H^?}{s'_1}{s_2}{ \psi}}{ \quadruple{H^?}{s_1}{s_2}{\psi}} \qquad
    \inferrule[Sim-R]{ (s'_2, s_2) \in \texttt{sim} \\ \quadruple{H^?}{s_1}{s'_2}{ \psi}}{ \quadruple{H^?}{s_1}{s_2}{\psi}}
  \end{mathpar}
  These rules can be used at any point in a proof to replace a state with a similar state of our choice.
  The soundness of \textsc{Sim-L} and \textsc{Sim-R} is immediately obtained from \textsc{Up-to} by reasoning up to \emph{simulation}.
  We define the corresponding (compatible) closure operator as follows: \[
    \texttt{simclo}(H) \triangleq \{ \ (s_1, s_2) \mid \exists (s'_1, s'_2), (s_1, s'_1) \in \texttt{sim} \wedge (s'_2, s_2) \in \texttt{sim} \wedge (s'_1, s'_2) \in H \ \}
  \]

  We refer the readers to our Coq development for a proof that \texttt{simclo} is compatible with \texttt{feF}.
  Since simulation is just a special instance of $\fe$ with $\psi = \LTLsquare \mathit{eq}$, we note that the simulation rules can be reformulated purely in terms of our semantic quadruples as follows: \begin{mathpar}
    \inferrule[Sim-L \textrm{(internalized)}]{ \quadruple{\fbox{$\emptyset$}}{s_1}{s'_1}{ \LTLsquare \mathit{eq}} \\ \quadruple{H^?}{s'_1}{s_2}{ \psi }}{ \quadruple{H^?}{s_1}{s_2}{\psi}} \quad
    \inferrule[Sim-R \textrm{(internalized)}]{ \quadruple{\fbox{$\emptyset$}}{s'_2}{s_2}{\LTLsquare\mathit{eq}} \\ \quadruple{H^?}{s_1}{s'_2}{ \psi }}{ \quadruple{H^?}{s_1}{s_2}{\psi}}
  \end{mathpar}

  \section{From $\forall\exists$ to $\forall^*\exists^*$}

  In the previous sections, we focused exclusively on hyperproperties with one universal quantifier followed by one existential quantifier. We note that our coinductive relations can naturally be extended to support the more general case of hyperproperties of the form $\forall^*\exists^*$ (i.e., $0$ or more universal quantifiers followed by $0$ or more existential quantifiers).
  Given a family of $m + n$ transition systems $\mathit{TS}_1, \ldots, \mathit{TS}_{n + m}$ such that $\mathit{TS}_i = (\mathcal{S}_i, \mathcal{E}_i, \mathcal{I}_i, \to_i)$ for all $0 \le i \le n + m$, we define the coinductive relation $\texttt{hyco}^{m, n} \subseteq \mathcal{S}_1 \times \ldots \times \mathcal{S}_{n + m} \times \mathcal{P}(\mathcal{E}^\omega_1 \times \ldots \times \mathcal{E}^\omega_{n + m})$ as follows: \begin{equation*}
    \begin{split}
      \texttt{hycoF}^{m, n}(R) \triangleq \{ \ 
      & (s_1, \ldots s_m, s_{m + 1}, \ldots, s_{m + n}) \mid \\
      &\forall \iosteps{s_1}{e_1}{s'_1} \ldots \forall\iosteps{s_m}{e_m}{s'_m},
      \exists \iosteps{s_{m + 1}}{e_{m + 1}}{s'_{m + 1}} \ldots \exists \iosteps{s_{m + n}}{e_{m + n}}{s'_{m + n}},\\
      & \Delta_{e_1, \ldots, e_{n + m}}(\psi) \ne \emptyset \wedge (s'_1, \ldots, s'_{m + n}, \Delta_{e_1, \ldots, e_{n + m}}(\psi)) \in R
      \ \}
    \end{split}
  \end{equation*}
  The associated greatest fixed point $\texttt{hyco}^{m, n} \triangleq \nu . \texttt{hycoF}^{m, n}$ gives a sound proof technique for $\forall^*\exists^*$ hyperproperties.

  \begin{theorem}
    Let $\psi \subseteq \mathcal{P}(\mathcal{E}^\omega_1 \times \ldots \times \mathcal{E}^\omega_{n + m})$ be a $(m + n)$-ary safety relation. Then, \[
      (\mathcal{I}_1, \ldots, \mathcal{I}_{m + n}, \psi) \in \texttt{hyco}^{m, n} \implies \models \forall^{\mathit{TS}_1}\ldots\forall^{\mathit{TS}_m}\exists^{{TS}_{m + 1}}\ldots\exists^{\mathit{TS}_{m + n}}\psi
    \]
  \end{theorem}

  \section{Proving Temporal Hyperproperties of Imperative Programs with I/O}

  So far, we have introduced a general framework to prove $\forall\exists$ hyperproperties with a safety relation between the universal and existential traces.
  In this section, we show that this framework naturally applies to the verification of imperative programs with inputs, outputs, and nondeterminism ($\textsc{IMP}_\mathit{io}$).

  \subsection{An Imperative Programming Language with Nondeterminism and I/O}

  $\textsc{IMP}_\mathit{io}$ programs support infinite loops, inputs, outputs, conditionals, and nondeterministic generation of constant values. The syntax is described in \Cref{fig:syntax}.

  \begin{figure}[H]
      \begin{center}
        \begin{tabular}{c}
          \\
          $P$ ::= \textbf{loop} $P$ | \textbf{if} $c$ \textbf{then} $P$ \textbf{else} $P$ | $\textbf{input}$ $x$ | \textbf{output} $e$ | $\textbf{havoc}$ $x$ | $P$ \textbf{;} $P$ | $x$ \textbf{:=} $e$\\
          \\
          $e$ ::= $x$ | $e$ \textbf{+} $e$ | $e$ \textbf{-} $e$ | $0$ | $1$ | $2$ | ... \qquad
          $c$ ::= \textbf{true} | \textbf{false} | e \textbf{<} $e$ | $e$ \textbf{=} $e$ | ...
        \end{tabular}
      \end{center}
      \caption{Syntax of $\textsc{IMP}_\mathit{io}$}
      \label{fig:syntax}
      \Description{Syntax of $\textsc{IMP}_\mathit{io}$}
  \end{figure}


  In \Cref{fig:syntax}, $x$ denotes a variable drawn from a set of program variables $\mathcal{X}$, $P$ denotes a program, $e$ an arithmetic expression, and $c$ a boolean expression.
  Given a memory $m : \mathcal{X} \to \mathbb{Z}$, we denote by $\llbracket e \rrbracket_m \in \mathbb{Z}$ (resp. $\llbracket c \rrbracket_m \in \{ \mathit{true}, \mathit{false} \}$) the value of arithmetic (resp. boolean) expressions.
  We interpret $\textsc{IMP}_\mathit{io}$ programs as labeled transition systems by giving a small-step operational semantics which, in turn, naturally induces a trace semantics.

  \begin{figure}[ht!]
    \begin{center}
      \begin{mathpar}
        \inferrule[Loop]{ }{\langle \textbf{loop} \ P, m \rangle \to \langle P ; \textbf{loop} \ P, m \rangle} \qquad
        \inferrule[Continue]{ }{\langle (P_1 \ ; \ P_2) \ ; \ P_3, m \rangle \to \langle P_1 \ ; \ (P_2 \ ; \ P_3), m \rangle}\\
        \inferrule[Input]{ v \in \mathbb{Z} }{\langle \textbf{input} \ x \ ; \ P, m \rangle \xrightarrow{\texttt{in}(v)} \langle P, m[x \gets v] \rangle} \qquad
        \inferrule[Output]{ v = \llbracket e \rrbracket_m }{\langle \textbf{output} \ e \ ; \ P, m \rangle \xrightarrow{\texttt{out}(v)} \langle P, m \rangle}\\
        \inferrule[Havoc]{ v \in \mathbb{Z} }{\langle \textbf{havoc} \ x \ ; \ P, m \rangle \to \langle P, m[x \gets v] \rangle} \qquad
        \inferrule[Assign]{ v = \llbracket e \rrbracket_m }{\langle x \ \textbf{:=} \ e \ ; \  P, m \rangle \to \langle P, m[x \gets v] \rangle}\\
        \inferrule[Ite-true]{ \llbracket c \rrbracket_m = \textbf{true} }{\langle \textbf{if} \ c \ \textbf{then} \ P_1 \ \textbf{else} \ P_2, m \rangle \to \langle P_1, m \rangle}\qquad \inferrule[Ite-false]{ \llbracket c \rrbracket_m = \textbf{false} }{\langle \textbf{if} \ c \ \textbf{then} \ P_1 \ \textbf{else} \ P_2, m \rangle \to \langle P_2, m \rangle}
      \end{mathpar}
    \end{center}
    \caption{Operational semantics of $\textsc{IMP}_\mathit{io}$}
    \Description{Operational semantics}
    \label{fig:semantics}
  \end{figure}

  We note that $\textbf{havoc}$ and $\textbf{input}$ have the same effect on the memory, but $\textbf{input}$ effects are recorded in program traces, whereas $\textbf{havoc}$ is used to model silent nondeterministic assignments and does not emit any event.
  Further, since we are considering only \textit{truly} reactive programs, we do not assign a meaning to basic instructions when they are not followed by more instructions (see rules $\textsc{Input}$, $\textsc{Output}$, $\textsc{Havoc}$, and $\textsc{Assign}$). Such programs would not have any infinite trace anyway.

  \subsection{A Deductive System for Hyperproperties of Imperative Programs}

  \newcommand{\hyco}[6]{ \quadruple{\fbox{\ensuremath{#6}}}{\ensuremath{\langle #1, #2\rangle}}{\ensuremath{\langle #3, #4 \rangle}}{\ensuremath{#5}}}
  \newcommand{\uhyco}[8]{ \quadruple{\fbox{\ensuremath{#6}}}{\ensuremath{#7 \triangleright \langle #1, #2 \rangle}}{\ensuremath{#8 \triangleright \langle #3, #4 \rangle}}{\ensuremath{#5}}}

  Let $\mathbb{P}$ be the set of all programs, $\mathcal{S} \triangleq \mathbb{P} \times \mathbb{Z}^\mathcal{X}$ be the set of all program states (pairs of a program and a memory), and $\mathcal{E} \triangleq \{ \texttt{in}(x) \mid x \in \mathbb{Z} \} \cup \{ \texttt{out}(x) \mid x \in \mathbb{Z} \}$ be the set of possible events emitted by $\mathit{IMP}_\mathit{io}$ programs. Any program $P$ can be viewed as labeled transition system $\mathit{TS}_P \triangleq (\mathcal{S}, \mathcal{E}, \mathcal{I}_P, \to)$ where $\to$ is the small-step operational semantics of $\textsc{IMP}_\mathit{io}$, and $\mathcal{I}_P \triangleq \langle P, \_ \mapsto 0 \rangle$
  is the state initiating the execution of $P$ in a zero-initialized memory.
  In the following, we will write $P$ for $\mathit{TS}_P$.
  In this setting, we can immediately use our coinductive relation $\fe$ to verify hyperproperties of programs.
  Nonetheless, instead of reasoning with the high-level rules presented in the previous sections, we can exploit the syntactic structure of programs to obtain more convenient reasoning rules as well as dedicated proof tactics. Figure \ref{fig:rules} presents a selection of such reasoning rules. Many more can be derived.

  \subsubsection*{Proof Management}

  The first rule, \textsc{Init}, exploits the soundness of $\fe$ (\Cref{thm:fe}) to initialize a proof by parameterized coinduction.
  The second rule, \textsc{Memory-Invariant}, is derived from \textsc{Invariant} and allows to extend the current coinduction hypothesis with a relation that should hold for every pairs of memories from the point of application on.
  Importantly, contrary to the \textsc{Invariant} rule (see \Cref{fig:fe_core}), where the extension of the coinduction hypothesis refers to states of the underlying LTSs, \textsc{Memory-Invariant} is restricted to invariants that only refer to the memories of the two programs (ignoring the code of the programs themselves, and the current temporal relation $\psi$). This is not a hard restriction (we could still use the more general rule), but we adopt this version for convenience as it is relatively unusual to allow invariants to refer to the syntactic structure of a program or to a logical specification.
  A memory-invariant $\mathit{INV}$ is lifted to a relation on triples in a natural way. More concretely, we define \[
    \mathit{INV}@(P_1, P_2, \psi) \triangleq \{ \ (\langle P_1, m_1 \rangle, \langle P_2, m_2 \rangle, \psi) \mid (m_1, m_2) \in \mathit{INV} \ \}
  \]
  The rule $\textsc{Memory-Invariant}$ extends the current coinduction hypothesis with $\mathit{INV}@(P_1, P_2, \psi)$, allowing to finish a proof at a later point if we cycle back to a state where the same pair of programs has to be matched with the same temporal relation $\psi$.

  \subsubsection*{Handling I/O}

  To handle inputs and outputs, we introduce two specialized rules \textsc{Input-input} and \textsc{output-output}. These rules are exploiting the fact that if the next instruction to be executed in both programs is an I/O operation, the only way to make progress is to ensure that any event produced by the left-hand program can be matched with a corresponding event in the right-hand program. Once the appropriate event for the right-hand program has been chosen, the new goal is guarded by a \texttt{next} operator, forcing us to compute the derivative of the current property $\psi$ and prove that it is not trivially violated.

  \subsubsection*{Handling Nondeterminism}
  
  Nondeterministic assignments are covered by the rules \textsc{Havoc-L} and \textsc{Havoc-R}.
  When the next instruction to execute in the left-hand program is a nondeterministic assignment of $x$, \textsc{Havoc-L} requires us to consider an arbitrary new value $v$ for $x$ (i.e., the new goal is guarded by a universal quantification over $v$).
  Dually, the rule $\textsc{Havoc-R}$ handles nondeterministic assignments in the right-hand program by requiring us to \textit{choose} a new value $v$ for $x$ (i.e., the new goal is guarded by an existential quantification over $v$).
  The rules \textsc{Havoc-R} and \textsc{Havoc-L} are derived from the more general \textsc{Steps-L} and \textsc{Steps-R} and from the operational semantics of $\textsc{IMP}_\mathit{io}$.

  \subsubsection*{Handling Loops}

  To handle loops, one could simply use the rules \textsc{Steps-L} and \textsc{Steps-R} to unfold a loop either in the right-hand or the left-hand program. However, we almost always wish to establish an invariant when a loop is encountered.
  The rules \textsc{Loop-L} and \textsc{Loop-R} combine an application of \textsc{Steps-L} (or \textsc{Steps-R}) with an application of \textsc{Memory-Invariant} to simultaneously unfold loops and establish a new memory-invariant.

  \subsubsection*{Handling Derivatives}
  
  We note that after applying I/O rules, the new goals are always of the form $\quadruple{\fbox{$H$}}{\ldots \triangleright \ldots }{\ldots \triangleright \ldots}{\psi}$, forcing us to derive the current relation $\psi$ to check that the emitted events are not immediately violating it.
  To do so, we use the rules for derivatives introduced in \Cref{fig:deriv-hyco}.

  \begin{figure}
      \begin{mathpar}
        \\
        \textbf{Proof Management}\\
        \inferrule[Init]{\psi \ \textrm{is a safety relation} \\ \quadruple{\fbox{$\emptyset$}}{\mathcal{I}_{P_1}}{\mathcal{I}_{P_2}}{\psi}}{\models \forall^{P_1}\exists^{P_2}\psi}\\
        \inferrule[Memory-Invariant]{(m_1, m_2) \in \mathit{INV}\\\forall (m_1, m_2) \in \mathit{INV}, \hyco{P_1}{m_1}{P_2}{m_2}{\psi}{H \cup \mathit{INV}@(P_1, P_2, \psi)}}{\hyco{P_1}{m_1}{P_2}{m_2}{\psi}{H}}\\
        \\
        \textbf{I/O}\\
        \inferrule[Input-input]{\forall v_1, \exists v_2, \uhyco{P_1}{m_1[x_1 \mapsto v_1]}{P_2}{m_2[x_2 \mapsto v_2]}{\psi}{H}{\texttt{in}(v_1)}{\texttt{in}(v_2)}}{\hyco{\textbf{input} \ x_1 \ ; \ P_1}{m_1}{\textbf{input} \ x_2 \ ; \ P_2}{m_2}{\psi}{H}}\\
        \inferrule[Output-output]{\uhyco{P_1}{m_1}{P_2}{m_2}{\psi}{H}{\texttt{out}(\llbracket e_1 \rrbracket_{m_1})}{\texttt{out}(\llbracket e_2 \rrbracket_{m_2})}}{\hyco{\textbf{output} \ e_1 \ ; \ P_1}{m_1}{\textbf{output} \ e_2 \ ; \ P_2}{m_2}{\psi}{H}}\\
        \\
        \textbf{Nondeterminism}\\
        \inferrule[Havoc-L]{\forall v, \hyco{P_1}{m_1[x \mapsto v]}{P_2}{m_2}{\psi}{H}}{\hyco{\textbf{havoc} \ x \ ; \ P_1}{m_1}{P_2}{m_2}{\psi}{H}}\quad
        \inferrule[Havoc-R]{\exists v, \hyco{P_1}{m_1}{P_2}{m_2[x \mapsto v]}{\psi}{H}}{\hyco{P_1}{m_1}{\textbf{havoc} \ x \ ; \ P_2}{m_2}{\psi}{H}}\\
        \\
        \textbf{Loops}\\
        \inferrule[Loop-L]{(m_1, m_2) \in \mathit{INV} \\\\ \forall (m_1, m_2) \in \mathit{INV}, \hyco{P_1 \ ; \ \textbf{loop} \ P_1}{m_1}{P_2}{m_2}{\psi}{H \cup \mathit{INV}@(\textbf{loop} \ P_1, P_2, \psi)}}{\hyco{\textbf{loop} \ P_1}{m_1}{P_2}{m_2}{\psi}{H}}\\
        \inferrule[Loop-R]{(m_1, m_2) \in \mathit{INV} \\\\ \forall (m_1, m_2) \in \mathit{INV}, \hyco{P_1 \ ; \ \textbf{loop} \ P_1}{m_1}{P_2}{m_2}{\psi}{H \cup \mathit{INV}@(P_1, \textbf{loop} \ P_2, \psi)}}{\hyco{\textbf{loop} \ P_1}{m_1}{P_2}{m_2}{\psi}{H}}
        \\
      \end{mathpar}
      \caption{Selection of proof rules for $\textsc{IMP}_\mathit{io}$}
      \label{fig:rules}
      \Description{Selection of proof rules for}
  \end{figure}

  \newpage
  
  \subsection{Examples in Coq}

  \begin{example}
  
    Consider the following example of a simple echo server:

    \begin{center}
      \texttt{echo}:
      \begin{tabular}{l}
        \textbf{loop}\\
        \quad \textbf{input} \ $x$\\
        \quad \textbf{output} \ $x$\\
      \end{tabular}
    \end{center}

    We use HyCo to prove that for all executions of \texttt{echo}, there exists another execution such that the outputs of the second execution are always exactly the double of the outputs in the first execution. The inputs are never compared. Formally, we prove $\forall^\texttt{echo}\exists^\texttt{echo}\LTLglobally\mathit{double}$ where $\mathit{double} \triangleq \{ \ (e_1, e_2) \mid e_1 = \texttt{out}(x) \implies e_2 = \texttt{out}(2 * x) \ \}$.
    The proof script (slightly simplified for readability) corresponding to the proof of this example is the following: 
  \begin{center}
    \begin{BVerbatim}
Proof.
  hyco_init. hyco_sync. hyco_step.
  intros v1. exists (2 * v1).
  hyco_deriv. hyco_step.
    ...
  hyco_cycle.
Qed.
    \end{BVerbatim}
  \end{center}

    In this script, the tactic \verb|hyco_init| exploits the rule \textsc{Init} to initialize a proof by coinduction. It also implicitly uses the rule \textsc{Invariant} to add the initial states to the current coinduction hypothesis.
    The tactic \verb|hyco_sync| then steps through the execution of the two programs until both reach an I/O instruction. In our case, the programs are executed until their first \textbf{input} instruction.
    The tactic \verb|hyco_step| determines which I/O rule should be applied to make further progress in the coinductive proof. Gere, \textsc{Input-input} is selected and it remains to match any input $v_1$ with some input $v_2$. We pick $v_2 = 2 * v_1$.
    After choosing an appropriate input, we need to check that this choice is compatible with the current trace relation by computing its derivative.
    The tactic \verb|hyco_deriv| selects the appropriate derivation rule (see \Cref{fig:deriv-hyco}).
    Here, the rule \textsc{Deriv-$\LTLsquare$} is applied and requires us to prove that the two outputs $\texttt{out}(v_1)$ and $\texttt{out}(2 * v_1)$ are satisfying the event-invariant $\mathit{double}$. This fact is trivially proven by unfolding the definition of $\mathit{double}$ and reading the current content of the memories.
    Finally, the tactic \verb|hyco_cycle| terminates the proof by applying the coinduction hypothesis.
    Indeed, after completing one iteration of the loop, the two programs remaining to be executed are again two copies of \texttt{echo}, and the property we need to verify is still $\LTLsquare\mathit{double}$; we cycled back to the initial state of the proof.
    We note that the repeated use of the tactics \verb|hyco_sync| and \verb|hyco_step| could be partially automated, allowing us to focus only on the \textit{interesting} part of the proof: picking the appropriate input for the right-hand program.
  \end{example}

  \begin{example}
  We describe a second example that requires us to use the rule \textsc{Invariant} as well as the alignment rules discussed in \Cref{sec:align}. We consider the two following programs \verb|incr| and \verb|ndet_add|: \begin{center}
    \verb|incr|:
    \begin{tabular}{l}
      $x$ \textbf{:=} 0\\
      \textbf{loop}\\
      \quad $x$ \textbf{:=} $x$ + 1\\
      \quad \textbf{output} \ $x$\\
    \end{tabular}
    \qquad
    \verb|ndet_add|:
    \begin{tabular}{l}
      $x$ \textbf{:=} 0\\
      \textbf{loop}\\
      \quad \textbf{havoc} \ $y$\\
      \quad $x$ \textbf{:=} $x + y$\\
      \quad \textbf{output} \ $x$\\
    \end{tabular}
  \end{center}

  As for the previous example, we wish to verify $\forall^\texttt{incr}\exists^\texttt{ndet\_add}\LTLglobally\mathit{double}$. The corresponding proof script is presented below.

  \begin{center}
    \begin{BVerbatim}
Proof.
  hyco_init.
  hyco_left 2. hyco_right 2.
  hyco_invariant (fun m1 m2 => 2 * m1 "x" = m2 "x").
    ...
  hyco_left 3. hyco_right 3.
  apply (hyco_havoc_r 2). hyco_right 2.
  hyco_step.
  ... rewrite <- INV. ...
  hyco_cycle.
  ... rewrite <- INV. ...
Qed.
    \end{BVerbatim}
  \end{center}

  The proof starts by executing both programs until the beginning of the loop is reached. To do so, we use the tactics \verb|hyco_left| and \verb|hyco_right|.
  These tactics exploit the rules \textsc{Steps-L} and \textsc{Steps-R} (see section \ref{sec:align}) to perform $n$ silent computation steps in the left-hand or the right-hand program.
  Once the beginning of both loops is reached, we use the tactic \verb|hyco_invariant| to establish the memory-invariant $x_2 = 2 * x_1$. After proving that the invariant is satisfied initially, we again step through the programs and we resolve the nondeterministic assignment $\textbf{havoc} \ y$ from the right-hand program by choosing $y = 2$. Note that at this point of the proof, the left-hand program is ready to emit an output while the right-hand program still needs to update its variable $x$.
  We use the tactic \verb|hyco_right| to align both programs, and the tactic \verb|hyco_step| to match the two outputs.
  We note that in order to be able to prove that the event-invariant $\LTLsquare\mathit{double}$ is preserved after emitting the outputs, we critically need to use the memory-invariant (this corresponds to the instruction \verb|rewrite <- INV| in the script). Once we proved that the event-invariant is maintained, we can conclude the proof with an application of \verb|hyco_cycle|.

  \end{example}

  \section{Related Work and Discussion}
  
  \subsection{Game-Based Verification of Hyperliveness}

  Our approach is connected to a game-theoretic interpretation of hyperproperties introduced by Coenen et al. \cite{DBLP:conf/cav/CoenenFST19}. In their approach, the task of verifying a $\forall\exists$ property is viewed as a game between a \textit{universal player}, and an \textit{existential player}.
  The moves of the universal player correspond to transitions in the left-hand system, and the existential player is forced to answer with corresponding transitions in the right-hand system. The existential player wins the game if there is a \textit{strategy} to answer any move of the universal player without immediately violating the specification.
  In that case, the targeted hyperproperty holds.
  In the coalgebraic reading we presented in this paper, the greatest fixed point $\nu . \fef$ can be interpreted as the \textit{winning region} of the game: the set of game states from which the existential player is guaranteed to win.

  Originally, the game-based approach was introduced with the goal of automating the verification of $\forall\exists$ temporal hyperproperties. The key insight is that when the systems being verified are finite-state, the game arena can be effectively constructed and sent to an efficient game solver. In the case of infinite-state systems described by programs, the arena cannot be constructed explicitly and we need to resort to finite approximations in order to automatically solve the game.
  For example, Beutner and Finkbeiner~\cite{DBLP:conf/cav/BeutnerF22} used predicate abstraction to construct a finite approximation of the game.

  An important advantage of the coalgebraic approach over the game-based approach is that it does not require us to explicitly construct an approximation of the the game.
  Instead, we implicitly define the \textit{exact} wining region as a coinductive relation. This allows us to leverage parameterized coinduction as a deductive system to prove that the existential player has a winning strategy.
  Furthermore, by embedding our coalgebraic approach in a proof assistant such as Coq, we immediately benefit from its rich logic and its ecosystem of mathematical libraries to reason about the game.

  \subsection{Program Logics for Hyperproperties}

  There is a long tradition of verifying relational properties of sequential programs using extensions of Hoare logics (cf.~\cite{10.1007/978-3-030-61470-6_7}). \emph{Relational Hoare Logic} (RHL) was originally introduced by Benton~\cite{Benton04} and later extended to higher-order programs~\cite{AguirreBGGS19}, separation logic~\cite{Yang07}, probabilistic computations~\cite{BartheGB09}, and quantum computations~\cite{Unruh19,BartheHYYZ20}. D'Osualdo et al. introduced \emph{hyper-triples}~\cite{10.1145/3563298} as a unifying compositional building block for proofs of $k$-hypersafety properties. \emph{Cartesian hoare logic}~\cite{10.1145/2980983.2908092} is a sound and relatively complete calculus for $k$-safety properties.

  Most relational logics are restricted to properties that express a condition over a given set of programs or a $k$-fold self-composition of some program for some fixed $k$. Some extensions that go beyond $k$-safety are directed at specific properties such as differential privacy~\cite{BartheKOB13} and sensitivity~\cite{BartheEGHS18}. There are also several extensions that provide more general support for $\forall\exists$ hyperproperties. Dardinier and M\"uller~\cite{10.1145/3656437} introduced \emph{Hyper Hoare Logic}, a generalization of Hoare logic that lifts assertions to properties of arbitrary sets of states. Hyper Hoare Logic can reason about both the absence and the existence of combinations of executions. Other extensions to $\forall\exists$ hyperproperties include \emph{Forall-Exist Hoare Tuples} (FEHT)~\cite{DBLP:conf/tacas/Beutner24}, \emph{refinement quadruples}~\cite{10.1007/978-3-642-35722-0_3} and \emph{RHLE triples}~\cite{10.1007/978-3-031-21037-2_4}.
  Some of these approaches have been mechanized in a proof assistant (cf.~\cite{10.1145/3656437}). However, all these approaches are limited to the analysis of pre-post style relational specifications. Unlike our approach, they are thus are not well-suited to reason about reactive systems whose executions are inherently infinite. By contrast, the coalgebraic approach goes beyond pre-post style specifications to reason about infinite sequences of events.

  \subsection{Mechanized Frameworks for Simulation Proofs}

  Interaction Trees (ITrees)~\cite{itrees} have recently been introduced as a general-purpose coinductive structure to model potentially non-terminating computations within the Coq proof assistant.
  Several approaches have been developed to establish trace inclusion (resp. equivalence) between ITrees using simulation (resp. bisimulation) techniques \cite{gpaco, stuttering_for_free}.
  Similar to the framework presented in this paper, these approaches are based on variants of parameterized coinduction.
  However, they do not support temporal reasoning. An interesting and natural direction would be to apply our approach to the verification of temporal hyperproperties of programs modeled as ITrees.

  Another Coq-based framework for simulation proofs is \textsc{Simuliris} \cite{simuliris}.
  It combines \textsc{Iris} \cite{iris}, a powerful concurrent separation logic, with simulation techniques.
  While \textsc{Simuliris} does not target temporal hyperproperties, integrating our framework within a separation logic in a similar way is an interesting research direction. It would enable intuitive reasoning about temporal hyperproperties of heap-manipulating programs.


  \section{Conclusion}

We have presented \textsc{HyCo}, a mechanized framework for the verification of
temporal hyperproperties within the Coq proof assistant.
Our approach provides a foundation for the construction of trustworthy proofs of temporal hyperproperties in complex reactive systems.
Since we use the full logic of Coq as the underlying assertion language, the
expressiveness of the hyperproperties considered here is far beyond the scope of currently available automated approaches.
In particular, \textsc{HyCo} can easily be applied to new programming languages and new temporal logics.

In future work, we plan to build on our framework to design a fully-featured mechanized program logic for temporal hyperproperties of reactive systems. Another important direction is to investigate the completeness of the approach. It is well-known that the game-theoretic approach is incomplete in general~\cite{DBLP:conf/cav/CoenenFST19}. For finite-state systems, the problem has been mitigated by adding \textit{prophecy variables} that inform the existential player about future choices of the universal player~\cite{9919658}. Our proof system would likely benefit from the introduction of prophecy variables in a similar manner.

  \section*{Data Availability}

  The Coq development accompanying this paper is available on GitHub at the following address: \begin{center}
    \href{https://github.com/acorrenson/hyco-popl-2025}{https://github.com/acorrenson/hyco-popl-2025}
  \end{center}

  Additionally, an archive containing the development version at the time this paper was submitted, together with detailed installation and evaluation instructions, can be found on Zenodo at the following address \cite{correnson_2024_14055009}: \begin{center}
    \href{https://zenodo.org/records/14055009}{https://zenodo.org/records/14055009}
  \end{center}

  \section*{Acknowledgements}

  This work was supported by the European Research Council (ERC) Grant HYPER (No. 101055412). Views and opinions expressed are however those of the authors only and do not necessarily reflect those of the European Union or the European Research Council Executive Agency. Neither the European Union nor the granting authority can be held responsible for them. A. Correnson carried out this work as a member of the Saarbr\"ucken Graduate School of Computer Science.

  \bibliographystyle{ACM-Reference-Format}
  \bibliography{references,hyper}
\end{document}